\documentclass[journal,final,twocolumn,twoside]{IEEEtran}
\usepackage{amsmath,amsfonts}
\usepackage{algorithmic}
\usepackage{algorithm}
\usepackage{amssymb}
\usepackage{subfigure}
\usepackage{color}
\usepackage{threeparttable}

\allowdisplaybreaks[4]
\usepackage{float} 

\usepackage{array}
\usepackage[caption=false,font=normalsize,labelfont=sf,textfont=sf]{subfig}
\usepackage{textcomp}
\usepackage{stfloats}
\usepackage{url}
\usepackage{bm}
\usepackage{verbatim}
\usepackage{multirow} 
\usepackage{ntheorem}
\usepackage{graphicx}
\usepackage{cite}
\theorembodyfont{\upshape}
\theoremheaderfont{\it}
\qedsymbol{\square}
\theoremseparator{:}
\newtheorem{theorem}{\indent Theorem}
\newtheorem{lemma}{\indent Lemma}
\newtheorem*{proof}{\indent Proof}
\newtheorem{remark}{\indent Remark}
\newtheorem{corollary}{\indent Corollary}

\makeatletter

\newcommand{\Rmnum}[1]{\expandafter\@slowromancap\romannumeral #1@}
\makeatother
\begin{document}

\title{\huge Wireless Federated Learning over Resource-Constrained Networks: Digital versus Analog Transmissions}

\author{Jiacheng~Yao,~\IEEEmembership{Graduate Student~Member,~IEEE,}
        Wei~Xu,~\IEEEmembership{Senior~Member,~IEEE,}
        Zhaohui~Yang,~\IEEEmembership{Member,~IEEE,}
        Xiaohu~You,~\IEEEmembership{Fellow,~IEEE,}
        Mehdi~Bennis,~\IEEEmembership{Fellow,~IEEE,}
        and H. Vincent Poor,~\IEEEmembership{Life~Fellow,~IEEE}
%        % <-this % stops a space
\thanks{Part of this work is presented in \emph{IEEE ICC 2024} \cite{icc}.}
\thanks{J. Yao, W. Xu and X. You are with the National Mobile Communications Research Laboratory (NCRL), Southeast University, Nanjing 210096, China (\{jcyao, wxu\}@seu.edu.cn).}
\thanks{Zhaohui Yang is with the Zhejiang Lab, Hangzhou 311121, China, and also with the College of Information Science and Electronic Engineering, Zhejiang University, Hangzhou, Zhejiang 310027, China (yang\_zhaohui@zju.edu.cn).}
\thanks{Mehdi Bennis is with the Center for Wireless Communications, Oulu
University, Oulu 90014, Finland (e-mail: mehdi.bennis@oulu.fi).}
\thanks{H. Vincent Poor is with the Department of Electrical and Computer Engineering, Princeton University, NJ 08544 USA (e-mail: poor@princeton.edu).}
}


\maketitle

\begin{abstract}
To enable wireless federated learning (FL) in communication resource-constrained networks, two communication schemes, i.e., digital and analog ones, are effective solutions. In this paper, we quantitatively compare these two techniques, highlighting their essential differences as well as respectively suitable scenarios. We first examine both digital and analog transmission schemes, together with a unified and fair comparison framework under imbalanced device sampling, strict latency targets, and transmit power constraints. A universal convergence analysis under various imperfections is established for evaluating the performance of FL over wireless networks. These analytical results reveal that the fundamental difference between the digital and analog communications lies in whether communication and computation are jointly designed or not. The digital scheme decouples the communication design from FL computing tasks, making it difficult to support uplink transmission from massive devices with limited bandwidth and hence the performance is mainly communication-limited. In contrast, the analog communication allows over-the-air computation (AirComp) and achieves better spectrum utilization. However, the computation-oriented analog transmission reduces power efficiency, and its performance is sensitive to computation errors from imperfect channel state information (CSI). Furthermore, device sampling for both schemes are optimized and differences in sampling optimization are analyzed. Numerical results verify the theoretical analysis and affirm the superior performance of  the sampling optimization.
\end{abstract}

\begin{IEEEkeywords}
Federated learning (FL), digital communication, over-the-air computation (AirComp), convergence analysis.
\end{IEEEkeywords}

\section{Introduction}
\IEEEPARstart{T}{he} dramatic development of data science has catalyzed significant advances in artificial intelligence (AI), which is driving innovation for anticipated sixth-generation (6G) mobile networks. The integration of AI and communication is envisioned to drive the shift from connected things to ubiquitous connected intelligence in wireless networks, supporting a large number of emerging intelligent applications \cite{6g1,6g2,ris2,xujindan}. 
Nonetheless, traditional centralized learning paradigms depend on extensive data transmission and considerable computational resources at cloud servers, which is challenging to implement in wireless networks. To better embrace AI, edge learning (EL) is viewed as a promising distributed learning technique that harnesses massive data and computational capacity available in edge devices distributed across wireless networks \cite{xu,pushai,mzchen}. Distinguishing it from the traditional separate design~for computation and communication, EL integrates the two and achieves efficient utilization of resources and improves performance through learning task-oriented communication design.

In particular, a key EL paradigm, namely federated learning (FL), has garnered significant attention from both academic and industrial circles, primarily due to its communication-efficient and privacy-enhancing characteristics \cite{fl1, fl2}. In FL, distributed edge devices utilize local datasets to collaboratively train a shared learning model with the assistance of a central parameter server (PS). By exchanging model parameters instead of raw data, the PS iteratively updates the global model until convergence. FL scheme minimizes the amount of transmitted data, as well as helping safeguard privacy and security. Recent studies have explored implementation of FL algorithms at wireless edge to support emerging AI applications \cite{lsfan1, wfl1, energy, wf2}. However, limited communication resources has posed a significant bottleneck to the performance of wireless FL \cite{sp,mz2}. One particular concern regards the uplink transmission process, where numerous participating devices need to transmit local updates to the PS, leading to a substantial increase in communication overhead and transmission latency \cite{mz3}. Hence, the development of efficient uplink transmission is crucial to enable wireless FL.

To support data transmission in wireless FL, digital communication schemes have been widely considered in recent works, where local updates are quantized into finite bits and then transmitted to the PS via traditional frequency division multiple access (FDMA) and time division multiple access (TDMA) schemes. At the receiver, the PS relies on channel coding for error detection and correction, before model aggregation using the received local updates. In \cite{wfl1} and \cite{unreliable}, the authors characterized the impact of packet errors on the convergence of FL, which enabled a task-oriented communication resource allocation scheme. The influence of various finite-precision quantization schemes in uplink and downlink communications was considered in \cite{quan1}. Building upon convergence analysis of the quantized FL, the quantization bits allocation was optimized in \cite{quan2} and \cite{quan3} to adapt channel diversity and requirements of the FL tasks. To further alleviate the communication bottleneck, one-bit quantization technique and reconfigurable intelligent surface (RIS) were used in \cite{onebit_RIS} to reduce communication overhead and enhance communication reliability, respectively.
Apart from resource allocation methods, modifications from the algorithmic perspective have been considered to combat unreliable transmissions. In \cite{ye}, the authors proposed a user datagram protocol (UDP)-based robust training algorithm, which  asymptotically achieved the same convergence rate as that with error-free communications. Moreover in \cite{gomore}, for replacing erroneous local updates, a global model reusing scheme, namely the GoMORE scheme, was devised to successfully mitigate the negative impacts of packet loss.  Alternatively, another solution is to further squeeze the communication overhead, thus improving the convergence over resource-constrained networks. The model pruning in \cite{pruning} was seen to be an effective way to compress the large-scale model into a smaller size, facilitating communication-efficient FL design.

In addition to these digital communication schemes, analog communication is an alternative communication-efficient way for deploying wireless FL. In particular, the local updates are amplitude-modulated and then simultaneously transmitted by reusing the available radio resource. Due to the superposition property of radio channels, the global model can be computed automatically over-the-air, which is therefore referred to as over-the-air computation (AirComp) \cite{poor}. Unlike the digital paradigm, analog communication pushes model aggregation from the PS to the air, which not only functionally but physically integrates the computation and communication. Benefiting from the over-the-air aggregation, the communication latency is substantially reduced  and the spectrum utilization is much more efficient, leading to fast-convergent and communication-efficient FL. It was shown in \cite{analog1} that the convergence rate of centralized learning remains approachable with this analog approach without power control and beamforming. Furthermore in \cite{zhu}, to combat deep fading, a novel truncated channel inversion scheme was proposed to exclude devices experiencing deep fades from the training process avoiding excessive energy consumption. Further insights into analog aggregation schemes were also discussed in the context of fundamental trade-offs between communication and learning. Besides, the impact of over-the-air aggregation errors on  optimality gap was analyzed in \cite{zhu2} and \cite{powercontrol} with power control optimization. Furthermore, the authors in \cite{jointpower} proposed an AirComp-based adaptive reweighing scheme for the aggregation, and jointly considered the power control and device selection deign based on the derived optimality gap. To combat the additive noise, robust FL training methods were proposed in \cite{robust} for both the expectation-based and the worst-case noise models. Considering multi-antenna scenarios, the beamforming design at the receiver was optimized by solving a sparse and low-rank optimization problem in \cite{shi}. In practice, considering the lack of perfect channel state information (CSI) for accurate power control, the work \cite{imperfect} investigated the impact of CSI uncertainty at the transmitter on FL convergence and revealed that CSI imperfection plays an key factor affecting the AirComp performance and convergence.

As mentioned above, by incorporating learning task-oriented resource allocation, both digital and analog transmissions are effective ways to fulfill the communication requirements of wireless FL \cite{haoy,wei,privacy}. 
In traditional communication for data transmission, digital communication schemes have been proven not only in theory but also in practice as dominantly outperforming analog communication techniques in almost all cases of interest. In communications for computation tasks, however, analog communication has shown  to be exceptionally effective in some cases of resource-constrained networks \cite{comp}. Hence, it is of interest to comprehensively compare digital and analog transmissions for wireless FL. Several recent studies have compared the two communication paradigms from some specific perspectives, including communication latency \cite{zhu,deploy} and convergence performance \cite{wf3,d2d}. However, to the best of our knowledge, there is a lack of literature that presents a comprehensive and quantitative comparison between the two fundamental communication paradigms, especially under practical constraints. Also, there have been few attempts to elucidate the fundamental differences between digital and analog transmissions in the context of FL, which is crucial for its deployment and design.

Against this background, in this paper, we conduct a theoretical comparison between the digital and analog transmission schemes under practical constraints.
The main contributions of this paper are summarized as follows.
\begin{itemize} 
\item We propose a unified framework for digital and analog transmissions in wireless FL, and characterize the model aggregation distortion caused by wireless transmission schemes. Using this framework, a fair comparison is conducted under the consideration of a stringent transmission delay target and two types of transmit power budgets. We exploit optimality gap, defined by the gap between the optimal and actually achieved loss function value, to characterize the convergence behavior and establish a stringent upper bound of the optimality gap for precise analysis and optimization in the digital/analog transmission enabled wireless FL.
It offers a precise characterization of the influence of wireless transmission imperfections on convergence in closed-form. 

%It captures the impacts of various factors, including additional noise, partial participation, coefficient distortion, and non-independent and non-identically distributed (non-IID) datasets.

\item Analytical results reveal that the digital transmission is hard to achieve satisfactory performance especially with limited radio resources due to orthogonal access and decoupled design. In contrast, the analog scheme exhibits a performance gain in terms of the optimality gap of the order of $\frac{1}{N}$ with the increasing number of participating devices, $N$, and thereby achieving a higher level of efficiency in spectrum utilization. However, the introduction of computation goals in the analog communication process results in less efficient transmit power utilization, and the presence of CSI uncertainties inevitably comes with computational distortion, thus enlarging the optimality gap by the order of $\frac{1}{\rho^2}$ with a decreasing level of channel estimation accuracy $\rho$.

\item Based on the derived optimality gap, we formulate an inclusion probability optimization problem for effective device sampling in wirless FL. The optimization problems for both digital and analog cases are optimally solved by checking the Karush-Kuhn-Tucker (KKT) conditions and exploiting the Dinkelbach algorithm, respectively. Through the examination of optimal solutions, we identify the essential differences underlying the device sampling optimization for digital and analog transmissions.

%\item  It is found crucial to optimize the inclusion probability for the digital transmission, in which  a 12~dB gain in terms of the transmit power is achieved.
\end{itemize}

Extensive numerical simulations are conducted to validate the derived analytical observations and the proposed sampling optimization. In particular, it is observed that the digital scheme has better power utilization, while the analog transmission is more spectrum-efficient.
 
The rest of this paper is organized as follows. In Section~\Rmnum{2}, we describe the typical FL algorithm, with details of digital and analog transmissions, and propose a fair comparison framework. Section \Rmnum{3} provides some preliminaries for the convergence analysis. In Section~\Rmnum{4}, we analyze the convergence performance under different transmission schemes and offer engineering insights. Then, in Section~\Rmnum{5}, we optimize the inclusion probabilities for both the digital and analog schemes. Simulation results and conclusions are given in Sections \Rmnum{6} and \Rmnum{7}, respectively.

\emph{Notation:} Boldface lowercase (uppercase) letters represent vectors (matrices). The set of all real numbers is denoted by $\mathbb{R}$. Superscripts $(\cdot)^T$ and $(\cdot)^\ast$ stand for the transpose and conjugate operations, respectively. The operator $\Re(\cdot)$ returns the real part of the input complex number. The operator $\left\|\cdot\right\|$ takes the Euclidean norm of vectors. A circularly symmetric complex Gaussian distribution is denoted by ${\cal {CN}}$, and $\mathbb{E}\{\cdot\}$ is the expectation operation.

\section{System Model and Communication Framework}

We consider a typical wireless FL system as shown in Fig.~\ref{fig1}, where $K$ distributed devices are coordinated by a central PS to perform FL. The training procedure and transmission model are elaborated in the sequel.

\begin{figure}[!t]
\centering
\includegraphics[width=2.8in]{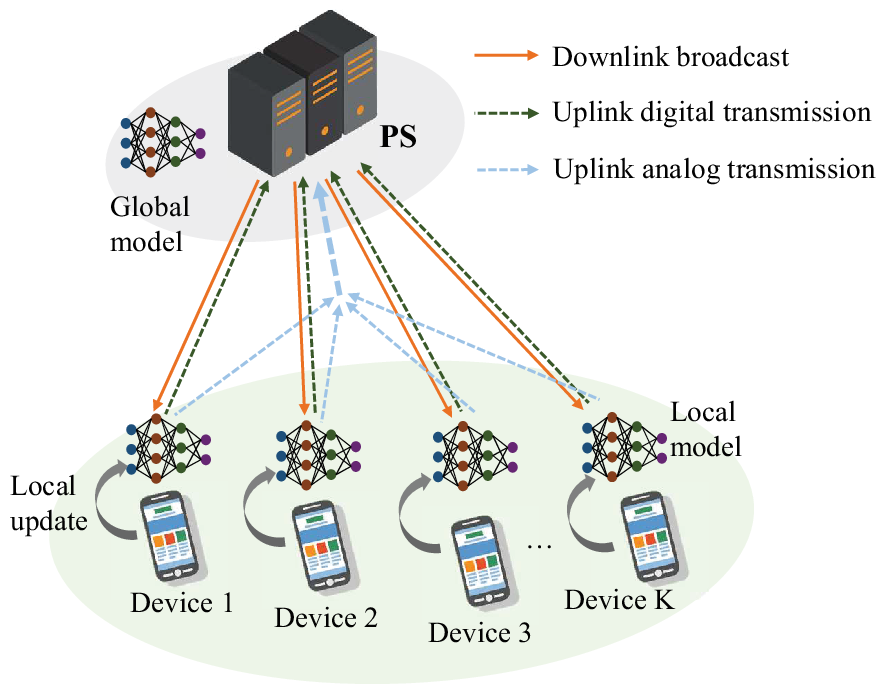}
\caption{ The architecture of a typical wireless FL system.}
\label{fig1}
\vspace{-0.35cm}
\end{figure}

\subsection{Federated Learning Model}
In FL, the distributed devices collaboratively train a shared machine learning model via local computing based on their local datasets and information exchange with the PS.
Let $\mathcal{D}_k$ denote the local dataset owned by the $k$-th device, which contains $D_k=\vert \mathcal{D}_k \vert$ training samples. The goal of the FL algorithm is to find the optimal $d$-dimensional model parameter vector, denoted by $\mathbf{w}^*\in \mathbb{R}^{d\times 1}$, to minimize the global loss function $F(\mathbf{w})$, i.e., 
\begin{align}\label{eq1}
\mathbf{w}^* &=\arg \min_{\mathbf{w}} F(\mathbf{w})=\arg \min_{\mathbf{w}}  \frac{1}{D} \sum_{k=1}^K D_k F_k(\mathbf{w})\nonumber \\
&=\arg \min_{\mathbf{w}} \sum_{k=1}^K \alpha_k F_k(\mathbf{w}),
\end{align}
where $D\triangleq \sum_{k=1}^K D_k$, $\alpha_k \triangleq \frac{D_k}{D}$ represents the aggregation weight for the $k$-th user, and $F_k(\mathbf{w})$ is the local loss function at device $k$ defined as
\begin{align}
F_k (\mathbf{w})= \frac{1}{D_k} \sum_{\mathbf{u}\in \mathcal{D}_k}	\mathcal{L}(\mathbf{w},\mathbf{u}),
\end{align}
where $\mathbf{u}$ denotes a training sample selected from $\mathcal{D}_k$, and $\mathcal{L}(\mathbf{w},\mathbf{u})$ represents the sample-wise loss function with respect to $\mathbf{u}$. Due to the heterogeneity of the system, we note that local datasets at distinct devices are usually non-independent and non-identically distributed (non-IID), and the optimal model parameters in (\ref{eq1}) are not necessarily the optimal for local datasets. Let $\mathbf{w}_k^*$ denote the locally optimal model at device $k$, i.e., $\mathbf{w}_k^* =\arg \min_{\mathbf{w}} F_k(\mathbf{w})$. It is  usually different from the globally optimal $\mathbf{w}^*$ unless the local dataset $\mathcal{D}_k$  experiences the same distribution as the whole data population.

To effectively handle the optimization problem in (\ref{eq1}), an FL algorithm performs the 
model training in an iterative manner. Specifically, the $m$-th round of the FL algorithm consists of the following steps.
\begin{itemize}
\item [1)]\emph{Model Broadcasting}: The PS broadcasts the latest global model $\mathbf{w}_m$ to al devices. 
    \item [2)]\emph{Local Computing}: After receiving $\mathbf{w}_m$, each device exploits its local dataset to compute the local gradient as
\begin{align}\label{eq3}
\mathbf{g}_m^k\triangleq \nabla F_k(\mathbf{w}_{m})= \frac{1}{D_k} \sum_{\mathbf{u}\in \mathcal{D}_k}	\nabla \mathcal{L}(\mathbf{w}_m,\mathbf{u}), \enspace \forall k.
\end{align}
\item [3)]\emph{Local Update Uploading}: Each device reports its local gradient to the PS.
\item [4)]\emph{Model Aggregation}: Upon receiving all local gradients, the PS updates the global model according to
\begin{align}
\mathbf{w}_{m+1}=\mathbf{w}_m-\eta \mathbf{g}_m,
\end{align}
where $\eta$ is the learning rate and $\mathbf{g}_m$ is given by
\begin{align}\label{gradient}
\mathbf{g}_m \triangleq \sum_{k=1}^K \alpha_k \mathbf{g}_m^k.
\end{align}
\end{itemize}
The above steps iterate until a convergence condition is met.

Considering the potentially massive number of devices and limited resources in practice, only a subset of devices can participate in each round of the training. Let $\mathcal{S}_m$ denote the set of activated devices selected in the $m$-th communication round and $N=\vert \mathcal{S}_m\vert$ be the number of participating devices per round. 
Due to  imbalanced dataset sizes and data heterogeneity, we assume that the PS performs non-uniform device sampling without replacement to select the participating devices per round. Specifically, the devices are randomly selected one by one from the remaining unselected device set. Once the number of selected devices reaches $N$, the sampling process terminates. Denote the inclusion probability of the device $k$ as $r_k$, which represents the probability of device $k$ being sampled per round and satisfies $r_k\leq 1$, $\forall k$, and $\sum_{k=1}^K r_k=N$. Due to the non-IID nature of the data, misaligned inclusion probability may bias the global model away from the local optimum, thereby decelerating the convergence and causing performance loss. Hence, in the following sections, we focus on the performance evaluation under fixed inclusion probabilities and characterize the impact of device sampling for  wireless FL.

Also, in wireless FL, the parameter transmission in Steps 1) and 3) relies on wireless communication between the PS and devices, which comes with additional imperfection in the model training procedure. Considering a sufficient power budget at the PS, the downlink transmission is usually assumed error-free \cite{wfl1}. Otherwise, for uplink transmission with limited communication resources, additional errors are inevitable. Efficient transmission and resource allocation schemes need to be designed to alleviate this impact of wireless environment. 

\subsection{Uplink Transmission Method}
We rely on the wireless uplink transmission to provide an estimation of the actual gradient in (\ref{gradient}). Assume that the total uplink bandwidth $B$ can be divided into up to $M$ subbands, which supports orthogonal access for $M$ devices. Without loss of generality, a frequency non-selective block fading channel model is adopted, where the wireless channels remain unchanged within a communication round. Let $\bar{h}_k =d_k^{-\frac{\alpha}{2}}h_k$  be the channel between the $k$-th device and the PS, where $d_k$ denotes the distance between the PS and device $k$, $\alpha$ represents the large-scale path loss exponent, and $h_k$ represents the small-scale fading of the channel. Assume that the channels are independent Rayleigh fadings, i.e., $h_k\sim \mathcal{CN}(0,1)$.
In practice, perfect estimation of the small-scale fading of the channel is usually not available. Let $\hat{h}_k$ denote the estimated channel at device $k$. Then, we model the CSI imperfection of the small-scale fading as 
\begin{align}
h_k = \rho \hat{h}_k + \sqrt{1-\rho^2}v_k,
\end{align}
where $\rho\in(0,1]$ is the correlation coefficient between $h_k$ and $\hat{h}_k$ to reflect the level of channel estimation accuracy, and $v_k\sim \mathcal{CN}(0,1)$ is the channel estimation error independent of $\hat{h}_k$. In the following, we introduce two typical uplink transmission schemes, i.e., digital and analog transmissions.

\subsubsection{Digital Transmission Model}
In the digital transmission, the $N$ selected devices first quantize their local updates into a finite number of  $b$ bits and then simultaneously transmit the quantized local updates to the PS. Specifically, we assume that the local update $\mathbf{g}_m^k$ is quantized by the stochastic quantization method in \cite{quan2}. Denote the maximum and the minimum values of the modulus among all parameters in $\mathbf{g}_m^k$ by $g_{m,\max}^k$ and $g_{m,\min}^k$, respectively. Then, the interval $[g_{m,\min}^k, g_{m,\max}^k]$ is divided evenly into $2^b-1$ quantization intervals. The uniformly distributed knobs are denoted by $\tau_i=g_{m,\min}^k+\frac{g_{m,\max}^k-g_{m,\min}^k}{2^b-1}i$ for $i=0,\cdots,2^b-1$. Given $\vert x \vert \in [\tau_i,\tau_{i+1})$, the quantization function $\mathcal{Q}(x)$ is expressed as
\begin{align}\label{quan}
\mathcal{Q}(x)=\left\{ \begin{array}{cc}
\mathrm{sign}(x)\tau_i &\mathrm{w.p.} \enspace \frac{\tau_{i+1}-\vert x\vert}{\tau_{i+1}-\tau_i},\\
\mathrm{sign}(x)\tau_{i+1} &\mathrm{w.p.} \enspace \frac{\vert x\vert-\tau_i}{\tau_{i+1}-\tau_i},\\
\end{array} \right.
\end{align}
where $\mathrm{sign}(\cdot)$ represents the signum function and ``w.p." represents ``with probability." Exploiting the quantization function in (\ref{quan}), the local update $\mathbf{g}_m^k$ is quantized as $\mathcal{Q}\left(\mathbf{g}_m^k \right)\triangleq \left[\mathcal{Q}\left(g_{m,1}^k \right),\cdots,\mathcal{Q}\left(g_{m,d}^k \right)\right]^T$, which is transmitted to the PS. Note that the exact value of $g_{m,\max}^k$ and $g_{m,\min}^k$ need to be transmitted to the PS with sufficient precision to support effective recovery. Hence, the total number of bits needed for transmitting amounts to
\begin{align}
b_{\mathrm{total}}=d(b+1)+q,
\end{align}
where $q$ is the number of bits used to represent  $g_{m,\max}^k$ and $g_{m,\min}^k$, and the additional one bit is the sign bit.

During the uplink FL parameter report, transmission errors are inevitable due to the channel dynamics and limited communication resources. Without loss of generality, we adopt the typical FDMA technique as an example. Assume that $M\geq N$ and hence each device can occupy different subbands equally to avoid interference with each other.\footnote{We generally assume orthogonal access between different devices and refrain from specifying the particular multiple access design. Hence, the following analysis can be safely extended to orthogonal access scenarios like TDMA and orthogonal frequency division multiple access (OFDMA).}
Then, the channel capacity of device $k$ can be evaluated as
\begin{align}
C_k=B_k \log_2\left( 1+\frac{P_k\vert \bar{h}_k\vert^2}{B_k N_0} \right),
\end{align}
where $B_k$ is the bandwidth allocated to device $k$ and it is set to $\frac{B}{N}$, $P_k$ is the transmit power at device $k$, and $N_0$ is the noise power density. 

The transmission delay under the digital transmission is primarily influenced by stragglers, which refer to devices with poor channel conditions. To avoid the uncontrolled severe delay brought by stragglers, we assume that all the devices transmit the local updates at a fixed rate rather than a dynamic one based on instantaneous signal-to-noise ratio (SNR) levels. Hence, the use of a fixed-rate transmission acts as a truncation mechanism for stragglers. Additionally, for devices experiencing favorable channel conditions, it is more beneficial to transmit at a lower rate with enhanced transmission reliability.
The target transmission rate is denoted by $R=\frac{B}{N}\log_2(1+\theta)$, where $\theta$ is
a chosen constant. According to \cite{wfl1}, the transmission is assumed error-free if the transmission rate is no larger than the channel capacity. Hence, the probability of successful transmission at device $k$ is calculated as
\begin{align}\label{eq9}
p_k=\Pr\left\{ R\leq C_k \right \}=\mathrm{exp} \left (-\frac{BN_0 \theta}{2NP_k d_k^{-\alpha}} \right).
\end{align}
At the PS, a cyclic redundancy check (CRC) mechanism is applied to check the detected data such that erroneous local updates can be excluded from the model aggregation. Finally, the obtained estimate of the desired gradient in (\ref{gradient}) is given by
\begin{align}\label{estd}
\hat{\mathbf{g}}_{m,\text{D}}= \sum_{k=1}^K \frac{\chi_k \alpha_k \xi_{k,\text{D}}}{r_k}\mathcal{Q}(\mathbf{g}_m^k),
\end{align}
where $\chi_k$ is an indicator variable for the device selection, and $\xi_{k,\text{D}}$ represents distortion brought by packet loss. 
To be concrete, $\chi_k$ is 1 if $k\in \mathcal{S}_m$ and otherwise $\chi_k$ is 0. Considering the definition of the inclusion probability, we have $\mathbb{E}\left[\chi_k\right]=r_k\leq 1$, which decreases the desired expected aggregation coefficient for unbiased gradient estimation. In order to compensate for the impact of partial participation, we multiply the coefficient $\frac{1}{r_k}$ in (\ref{estd}), such that $\frac{1}{r_k}\mathbb{E}\left[\chi_k\right]=1$.  Analogously, the distortion $\xi_{k,\text{D}}$ is characterized by the probability in (\ref{eq9}) as
\begin{align}\label{distortiond}
\xi_{k,\text{D}}=\left \{  
\begin{array}{cl}
\frac{1}{p_k} &\text{w.p.}\enspace p_k,\\
0&\text{w.p.}\enspace 1-p_k,
\end{array}
\right.
\end{align}
to ensure $\mathbb{E}\left[ \xi_{k,\text{D}}\right]=1$. With the gradient estimate in (\ref{estd}), the global model updated at the $(m+1)$-th round equals to
\begin{align}
\tilde{\mathbf{w}}_{m+1}= \tilde{\mathbf{w}}_{m}-\eta\hat{\mathbf{g}}_{m,\text{D}},
\end{align}
where $\tilde{\mathbf{w}}_{m}$ denotes the model obtained at the previous round.

\subsubsection{Analog Transmission Model}
In the analog transmission with AirComp, selected devices simultaneously upload the uncoded analog signals of local gradients to the PS by fully reusing the time-frequency resource. A weighted summation of the local updates in (\ref{gradient}) can be achieved by exploiting channel pre-equalization and the waveform superposition nature of the wireless channel. In this study, we consider that the total bandwidth is constrained for fair comparison and all subbands are utilized for the transmission of identical parameters.  This is because the uncoded nature of the analog transmission diminishes its robustness, rendering it more vulnerable to interference and even the malicious attacks.\footnote{The derived results directly extend to the case of dividing bandwidth for distinct parameter transmission in broadband scenarios\cite{zhu}.
}
Specifically, the received signal at the PS is expressed as 

\begin{align}\label{eq13}
\mathbf{y}=\sum_{k=1}^K \chi_k\bar{h}_k  \beta_k \mathbf{g}_m^k +\mathbf{z}_m,
\end{align} 
where $\beta_k$ is the pre-processing factor at device $k$, and  $\mathbf{z}_m$ is additive white Gaussian noise following $\mathcal{CN}(\mathbf{0},BN_0\mathbf{I})$. To accurately estimate the desired gradient in (\ref{gradient}), the pre-processing factor $\beta_k$ should be adapted to the channel coefficient $\bar{h}_k$. Unlike the digital transmission, CSI is needed at the transmitter for the analog transmission. 
Channel pre-equalization is performed based on the CSI available at each device. For simplicity, we adopt the typical truncated channel inversion scheme to combat deep fades\cite{zhu}. It is expressed as
\begin{align}\label{eq15}
    \beta_k=\left \{  
\begin{array}{cl}
\frac{\zeta \lambda \alpha_k d_k^{\frac{\alpha}{2}} \hat{h}_k^*}{ r_k \vert \hat{h}_k \vert^2} & \vert \hat{h}_k \vert^2 \geq \gamma_{\mathrm{th}},\\
0&\vert \hat{h}_k \vert^2 <\gamma_{\mathrm{th}},
\end{array}
\right.
\end{align}
where $\gamma_{\mathrm{th}}$ is a predetermined power-cutoff threshold, $\zeta$ is a scaling factor for ensuring the transmit power constraint, and compensation coefficient $\lambda$ is selected to alleviate the impact of imperfect CSI \cite{imperfect}. Through the pre-processing in (\ref{eq15}), we aim to eliminate the influence of the uneven channel fading $\bar{h}_k$, and the inclusion probability $p_k$, thereby ensuring the unbiased gradient estimation.

At the receiver, the PS scales the real part of $\mathbf{y}$ in (\ref{eq13}) with $\frac{1}{\zeta}$ and obtain an estimate of the actual gradient in (\ref{gradient}). It yields
\begin{align} \label{ghata}
\hat{\mathbf{g}}_{m,\text{A}}=\sum_{k=1}^K \frac{\chi_k \alpha_k \xi_{k,\text{A}} }{r_k } \mathbf{g}_m^k+\bar{\mathbf{z}}_m,
\end{align}
where $\bar{\mathbf{z}}_m\triangleq \frac{\Re\{\mathbf{z}_m\}}{\zeta}$ is the equivalent noise, and $\xi_{k,\text{A}}$ denotes the distortion brought by the analog transmission with imperfect CSI. It follows
\begin{align}
\xi_{k,\text{A}} =\left \{  
\begin{array}{cl}
\frac{\lambda\Re\{h_k^* \hat{h}_k\}}{ \vert \hat{h}_k \vert^2} &  \text{w.p.} \enspace \mathrm{e}^{-\gamma_{\mathrm{th}}},\\
0&\text{w.p.}\enspace 1-\mathrm{e}^{-\gamma_{\mathrm{th}}}.
\end{array}
\right.
\end{align}
Similarly, the global model at the $(m+1)$-th round under the analog transmission is updated as
\begin{align}
\tilde{\mathbf{w}}_{m+1}= \tilde{\mathbf{w}}_{m}-\eta\hat{\mathbf{g}}_{m,\text{A}}.
\end{align}

\subsection{A Unified Framework for Wireless FL Comparison}
%The constraint $\mathrm{C}_1$ requires unbiasedness of the gradient estimation, which plays an essential role in achieving satisfactory performance with the non-IID datasets.
To minimize the optimality gap brought by imperfect uplink transmission, the overall FL task-oriented optimization over the wireless networks can be formulated as
\begin{align} \label{p1}
    \text{minimize} \quad&\mathbb{E}\left[F(\mathbf{w}_{m+1})\right]-F(\mathbf{w}^*)\nonumber \\
    \text{subject to}\quad  & \mathrm{C}_1:T\leq T_{\max},\nonumber \\
    & \mathrm{C}_2:P_k\leq P_{\max},\enspace \forall k,
\end{align}
where the expectation is taken over channel dynamics, $T$ represents uplink transmission delay per round, $T_{\max}$ and $P_{\max}$ denotes the maximum transmission delay target and the transmit power unit, respectively. 
Constraint $\mathrm{C}_1$ and $\mathrm{C}_2$ respectively represent the maximum transmission delay and maximum transmit power constraint in practice. Apart from the maximum power budget, another typical transmit power constraint is the average power budget \cite{zhu}, i.e., 
\begin{align}
    \bar{\mathrm{C}}_2:\mathbb{E}[P_k]\leq P_{\mathrm{ave}},\enspace \forall k,
\end{align}
where $P_{\mathrm{ave}}$ denotes the average power budget and limits the energy consumption during the uplink transmission process.

\begin{table}[!t]   
\renewcommand\arraystretch{1.2}
\begin{center}   
\caption{Main Differences Between the Two Paradigms}  
\vspace{-0.3cm}
\label{table:1} 
\begin{tabular}{|c|c|c|c|}   
\hline 
Paradigms & Gradient estimation & Transmission delay & Power budget\\
\hline 
Digital & (\ref{estd}) & (\ref{delayd}) & (\ref{powerd})\\
\hline
Analog & (\ref{ghata}) &  (\ref{delaya}) & (\ref{pmax}), (\ref{pave})\\
\hline
\end{tabular}   
\end{center}   
\vspace{-0.5cm}
\end{table}

For fair comparison between the two transmission paradigms, we measure the achievable objective value of the problem in (\ref{p1}) under the same transmission delay target and transmit power budget. Specific constraints for the two transmission paradigms are listed as follows, summarized in Table \ref{table:1}. 

For digital transmission, the transmission delay per communication round is calculated as
\begin{align}\label{delayd}
T_{\text{D}}=\frac{b_{\mathrm{total}}}{R}= \frac{Nd(b+1)}{B\log_2(1+\theta)},
\end{align}
where the evaluation holds with a sufficiently large model size $d$. Hence, constraint $\mathrm{C}_1$ is reformulated as 
\begin{align} \label{eq22}
    \frac{Nd(b+1)}{B\log_2(1+\theta)}\leq T_{\max} \Rightarrow \theta\geq 2^{\frac{Nd(b+1)}{BT_{\max}}}-1.
\end{align}
For constraint $\mathrm{C}_2$, due to its interference-free characteristic, full power transmission is optimal and hence the constraint is reformulated by
\begin{align}\label{powerd}
    P_k=P_{\max}, \enspace \forall k.
\end{align}
Also, with the average transmit power budget, we assume invariant transmit power over different communication rounds and have $P_k=P_{\mathrm{ave}}, \enspace \forall k$.

For analog transmission, according to \cite[Eq. (16)]{deploy}, the per-round delay follows
\begin{align} \label{delaya}
T_{\text{A}}=\frac{dM}{B},
\end{align}
which is a constant. For feasibility, we assume that the target $T_{\max}$ cannot be smaller than $T_{\text{A}}$. The maximum power constraint $\mathrm{C}_2$ is rewritten as
\begin{align}\label{pmax}
\max_{m,k}\left \{ \left \Vert \beta_k \mathbf{g}_m^k \right \Vert^2 \right \} \leq P_{\max},
\end{align}
for the analog transmission. Unlike the digital transmission, it is impossible to fully utilize the maximum power in analog transmission due to the need for channel pre-equalization.
On the other hand, the average power constraint $\bar{\mathrm{C}}_2$ follows
\begin{align}\label{pave}
\mathbb{E}\left [ \left \Vert \beta_k \mathbf{g}_m^k \right \Vert^2 \right ] \leq P_{\mathrm{ave}},
\end{align}
where the expectation is taken over the wireless channel dynamics and different communication rounds.

\section{Preliminaries}

To pave the way for performance analysis, this section provides necessary assumptions and lemmas about the learning algorithms and the transmission paradigms, which will be useful in the next section.

\subsection{Assumptions for Learning Algorithms}
To begin with, we make several common assumptions on the loss functions, which are widely used in FL studies like \cite{wfl1,zhu2,importance}.

\emph{Assumption 1}: The local loss functions $F_k(\cdot)$ are $\mu$-strongly convex for all devices, that is
\begin{align}\label{convex}
F_k(\mathbf{w})\geq F_k(\mathbf{v}) +\nabla F_k(\mathbf{v})^T (\mathbf{w}-\mathbf{v})+\frac{\mu}{2}\Vert \mathbf{w}-\mathbf{v}\Vert^2.
\end{align}

\emph{Assumption 2}: The local loss functions $F_k(\cdot)$ are differentiable and have $L$-Lipschitz gradients, which follows
\begin{align}
\Vert \nabla F_k(\mathbf{w})-  \nabla F_k(\mathbf{v})\Vert\leq L \Vert \mathbf{w}-\mathbf{v}\Vert,
\end{align}
and it is equivalent to
\begin{align}\label{smooth}
F_k(\mathbf{w})\leq  F_k(\mathbf{v}) +\nabla F_k(\mathbf{v})^T (\mathbf{w}-\mathbf{v})+\frac{L}{2}\Vert \mathbf{w}-\mathbf{v}\Vert^2.
\end{align}

\emph{Assumption 3}: In most practical applications, it is safe to assume that the sample-wise gradient is always upper bounded by a finite constant $\gamma$, i.e., 
\begin{align}
\left \Vert \nabla \mathcal{L}(\mathbf{w},\mathbf{u}) \right \Vert \leq \gamma.
\end{align}

\emph{Assumption 4}: The distance between the locally optimal model, $\mathbf{w}_k^*$, and the globally optimal model, $\mathbf{w}^*$, is uniformly bounded by  a finite constant $\delta$, i.e., 
\begin{align}
\Vert \mathbf{w}_k^*- \mathbf{w}^* \Vert\leq \delta.
\end{align}

%Based on \emph{Assumption 1} and \emph{Assumption 2}, it is easy to verify that the global loss function $F(\cdot)$ is also $\mu$-strongly convex and $L$-smooth and we omit the proof for simplicity.

%Moreover, to capture the degree of non-IID, we evaluate the distance between the locally optimal models and the globally optimal one in the following assumption, which directly reflects the heterogeneity of the data distribution \cite{importance}.

\subsection{Preliminary Lemmas}

We present lemmas regarding the strong convexity and Lipschitz smooth properties of the global loss function. 

\begin{lemma}
    With $\mu$-strongly convex and $L$-smooth local loss functions, the global loss function $F(\cdot)$ is also $\mu$-strongly convex and $L$-smooth.
\end{lemma}

\begin{proof}
    Recalling the definition of $F(\cdot)$ in (\ref{eq1}), with \emph{Assumptions 1-2}, it is easily verified that any linear combination of $\mu$-strongly convex and $L$-smooth local loss functions also satisfies (\ref{convex}) and (\ref{smooth}). The proof completes. \hfill $\square$
\end{proof}

We then provide the following lemma regarding the imperfection in digital and analog transmission paradigms.

\begin{lemma}
    Under the stochastic quantization and the proposed digital aggregation in (\ref{estd}), $\hat{\mathbf{g}}_{m,\text{D}}$ is an unbiased estimate of the actual gradient in (\ref{gradient}). For the considered analog paradigm in (\ref{ghata}), by choosing $\lambda = \frac{e^{\gamma_{\mathrm{th}}}}{\rho}$, the gradient estimate $\hat{\mathbf{g}}_{m,\text{A}}$ is also unbiased.
\end{lemma}

\begin{proof}
    Please refer to Appendix \ref{lemma2}. \hfill $\square$
\end{proof}

Although both the digital and analog transmissions achieve unbiased gradient estimations, there are fundamental differences in the distortion between the two paradigms. For the digital transmission, the distortion mainly lies in the gradients themselves, i.e., gradient quantization errors. On the other hand, due to the integration of communication and computation in AirComp, the analog transmission additionally suffers from distortion in coefficient aggregation, i.e., computation error, which is due to the CSI imperfection. This essential difference further discriminates the performances of digital and analog transmissions, which are elaborated in the next section.

\section{Comparison with Convergence Analysis}
In this section, we analyze the convergence performance under the digital and analog transmissions with the practical constraints for wireless FL. Based on the derived results, we further conduct quantitative comparisons between the two paradigms from various perspectives of view.

\subsection{Convergence under the Maximum Power Budget}

We characterize the convergence performance under different transmission paradigms in the following theorems.

\begin{figure*}[t]
\begin{align}
\mathbb{E}&\left[ F(\tilde{\mathbf{w}}_{m+1})\right]-F(\mathbf{w}^*) \leq \frac{L}{2}\left( 1-\eta \mu +2\eta^2 L^2g_{\text{D}}(\mathbf{r},b)\right)^{m+1} \mathbb{E}\left[\left \Vert \tilde{\mathbf{w}}_{0} -\mathbf{w}^*\right \Vert^2 \right] +\frac{\eta (L\phi(b) +2L^3\delta^2) g_{\text{D}}(\mathbf{r},b)}{2\mu-4\eta L^2 g_{\text{D}}(\mathbf{r},b)},\label{digital}\\
\mathbb{E}&\left[ F(\tilde{\mathbf{w}}_{m+1})\right]-F(\mathbf{w}^*)\leq \frac{L}{2}\left( 1-\eta \mu +2\eta^2 L^2g_{\text{A}}(\mathbf{r},\gamma_{\mathrm{th}})\right)^{m+1} \mathbb{E}\left[\left \Vert \tilde{\mathbf{w}}_{0} -\mathbf{w}^*\right \Vert^2 \right] +\frac{\eta \left(L\varphi(\mathbf{r},\gamma_{\mathrm{th}})+2L^3 \delta^2 g_{\text{A}}(\mathbf{r},\gamma_{\mathrm{th}}) \right)}{2\mu-4\eta L^2 g_{\text{A}}(\mathbf{r},\gamma_{\mathrm{th}})}. \label{analog}
\end{align}
\hrulefill
\vspace{-0.4cm}
\end{figure*}

\begin{theorem}[Digital Transmission]
For a fixed learning rate satisfying $\eta \leq \frac{\mu}{2L^2 g_{\text{D}}(\mathbf{r},b)}$, the optimality gap of the distributed gradient update in the $(m+1)$-th iteration of the digital transmission is equal to (\ref{digital}) at the top of the next page,
where $\phi(b)$ is a constant defined in Appendix \ref{theo1} regarding the quantization errors, $\mathbf{r}\triangleq [r_1,\cdots,r_K]^T$, and $g_{\text{D}}(\mathbf{r},b)\triangleq \sum_{k=1}^K \frac{\alpha_k}{p_k r_k}$. 
\end{theorem}
\begin{proof}
Please refer to Appendix \ref{theo1}. \hfill $\square$
\end{proof}

\begin{theorem}[Analog Transmission]
For a fixed learning rate satisfying $\eta \leq \frac{\mu}{2L^2 g_{\text{A}}(\mathbf{r},\gamma_{\mathrm{th}})}$, the optimality gap of the distributed gradient update in the $(m+1)$-th iteration of the analog transmission
is equal to (\ref{analog}) at the top of the next page,
where $g_{\text{A}}(\mathbf{r},\gamma_{\mathrm{th}})\triangleq \sum_{k=1}^K \frac{\alpha_k}{r_k}\left(e^{\gamma_{\mathrm{th}}}+\frac{(1-\rho^2)\mathrm{E}_1\left(\gamma_{\mathrm{th}} \right)e^{2\gamma_{\mathrm{th}}}}{2\rho^2}\right)-1$, $\mathrm{E}_1(x)\!\triangleq\!\int_{x}^{\infty}\frac{e^{-t}}{t} \mathrm{d}t$, and $\varphi(\mathbf{r},\gamma_{\mathrm{th}})\!\triangleq\!\frac{B N_0 \gamma^2 e^{2\gamma_{\mathrm{th}}}}{2P_{\max}\rho^2 \gamma_{\mathrm{th}}}\max_k\left \{\frac{\alpha_k^2}{r_k^2 }d_k^{\alpha}\right\}$.
\end{theorem}
\begin{proof}
Please refer to Appendix \ref{theo2}. \hfill $\square$
\end{proof}

From \emph{Theorems 1-2}, we find that the convergence rate mainly depends on the choice of the learning rate $\eta$, while the imperfections in transmission also have a certain impact. We conclude the following immediate observations on the convergence rate.

\begin{remark}
As observed in (\ref{digital}) and (\ref{analog}), the convergence performace of an FL algorithm is negatively related to $g_{\text{D}}(\mathbf{r},b)$ for digital transmission and to $g_{\text{A}}(\mathbf{r},\gamma_{\mathrm{th}})$ for analog transmission. We refer to $g_{\text{D}}(\mathbf{r},b)$ and $g_{\text{A}}(\mathbf{r},\gamma_{\mathrm{th}})$ as the \emph{virtual sum weight} for the digital and analog transmissions, respectively, which reflects the degree of hindrance to the convergence imposed by unequal sampling and vulnerable wireless communication. Under the ideal case, with full device participation and no transmission outage, the \emph{virtual sum weight} equals to 1, otherwise it is amplified by the imperfect characteristics. It is interesting to note that, for devices with more data samples, i.e., larger $\alpha_k$, the impact of imperfections is exaggerated. 
\end{remark}

\begin{remark}
Comparing  $g_{\text{D}}(\mathbf{r},b)$ and $g_{\text{A}}(\mathbf{r},\gamma_{\mathrm{th}})$, it can be seen that the vulnerability of digital transmission introduces additional heterogeneity, i.e., varying $p_k$, which does not exist in the analog paradigm. This is because outage probability in the digital case is determined by channel conditions and varying across different devices. On the other hand, due to the uniform truncation threshold, all participating devices enjoy the same truncation probability in the analog transmission. Hence, in design of inclusion probabilities $\mathbf{r}$ for the digital case, we need to adapt the inclusion probabilities to both dataset size and channel condition. By contrast, in the case of analog transmission, only the heterogeneity of the dataset size needs to be considered.
\end{remark}

%\begin{remark}
%It is worth noting that the convergence under digital transmission can be accelerated by improving SNR while under the analog transmission it does not accelerate the convergence. Consequently, the virtual sum weight tends to the ideal case with sufficiently large SNR, i.e., the best convergence rate is achievable. This is due to the fact that coefficient distortion in the imperfect AirComp cannot be mitigated by increasing transmit power under a given level of channel estimation accuracy. To improve the convergence rate, optimization with respect to the truncation threshold can be useful,  and the best convergence rate is never achievable as long as the channel estimation is imperfect, i.e., $\rho<1$.
%\end{remark}

According to \emph{Theorems 1-2},  we are ready to derive the optimality gap after convergence for further evaluation in the following corollary, which reflects the ultimately achievable performance of the wireless FL.

\begin{corollary}
With sufficient iterations, the optimality gap achieved by digital and analog transmissions, respectively, converges to
\begin{align}
G_{\text{D}}&=\frac{\eta (L \phi(b) +2L^3\delta^2) g_{\text{D}}(\mathbf{r},b)}{2\mu-4\eta L^2 g_{\text{D}}(\mathbf{r},b)},\label{Gd}\\
G_{\text{A}}&= \frac{\eta \left(L\varphi(\mathbf{r},\gamma_{\mathrm{th}})+2L^3 \delta^2 g_{\text{A}}(\mathbf{r},\gamma_{\mathrm{th}}) \right)}{2\mu-4\eta L^2 g_{\text{A}}(\mathbf{r},\gamma_{\mathrm{th}})}.\label{Ga}
\end{align}
\end{corollary}
\begin{proof}
Consider the digital transmission scenario with a sufficient number of iterations. We have
\begin{align}
  &\lim_{m\to \infty} \mathbb{E}\left[ F(\tilde{\mathbf{w}}_{m+1})\right]-F(\mathbf{w}^*)  \nonumber \\
  &\enspace \leq \lim_{m\to \infty} \frac{L}{2}\left( 1-\eta \mu +2\eta^2 L^2g_{\text{D}}(\mathbf{r},b)\right)^{m+1} \mathbb{E}\left[\left \Vert \tilde{\mathbf{w}}_{0} -\mathbf{w}^*\right \Vert^2 \right] \nonumber \\
  &\enspace \quad +\frac{\eta (L\phi(b) +2L^3\delta^2) g_{\text{D}}(\mathbf{r},b)}{2\mu-4\eta L^2 g_{\text{D}}(\mathbf{r},b)}\nonumber \\
  &\enspace \overset{\text{(a)}}{=}\frac{\eta (L\phi(b) +2L^3\delta^2) g_{\text{D}}(\mathbf{r},b)}{2\mu-4\eta L^2 g_{\text{D}}(\mathbf{r},b)}= G_{\text{D}},
\end{align}
where the inequality is obtained through \emph{Theorem 1} and the equality in (a) is due to the fact that  $\eta < \frac{\mu}{2L^2 g_{\text{D}}(\mathbf{r},b)}$, i.e., $\left( 1-\eta \mu +2\eta^2 L^2g_{\text{D}}(\mathbf{r},b)\right)<1$. Hence, the achieved optimality gap at convergence is bounded by $G_{\text{D}}$. As for the analog transmission, the proof is almost the same and is omitted here for simplity. \hfill $\square$
\end{proof}

%From \emph{Corollary 1}, we note that the optimality gap with digital transmission is made up of two aspects, namely quantization errors and partial participation. On the other hand, in analog transmission case, the optimality gap comes from the noise, imperfect CSI, and partial participation. 
From \emph{Corollary 1}, we further compare the two typical paradigms from the following perspectives and conclude insightful remarks that are instructive for the deployment of FL in wireless networks. As a summary, we list main comparison results in Table \ref{table:2}.
For the sake of simplicity in analysis, without loss of generality, we drop the unbalance of the datasets and assume uniform inclusion probabilities, i.e., $\alpha_k =\frac{1}{K}$, and $r_k=\frac{N}{K}$, $\forall k$, which does not cause any essential changes. Also we set that $T_{\max}=T_{\text{A}}$. Note that the learning rate is assumed to be sufficiently small and hence the convergence rate remains the same for all cases.

\begin{table*}[!t]   
\renewcommand\arraystretch{1.2}
\begin{center}   
\caption{Main Comparison Results with Respect To Optimality Gap}  
\label{table:2} 
\begin{tabular}{|m{1.5cm}<{\centering}|m{3cm}<{\centering}|m{3cm}<{\centering}|m{4.5cm}<{\centering}|m{3cm}<{\centering}|}   
\hline   
\multirow{2}*{Paradigms} &  \multicolumn{2}{c|}{Transmit power budget, $P$} & \multirow{2}*{Device number, $N$} & \multirow{2}*{Imperfect CSI, $\rho$} \\
\cline{2-3}
&Low SNR & High SNR & &\\
\hline 
Digital & $\mathcal{O}\left( \mathrm{exp}\left(\frac{\varepsilon}{P}\right)\right)$ $\searrow$ & $\rightarrow$ $G_{\text{D}}^{\infty}$ & $\mathcal{O}\left(\frac{1}{N}\mathrm{exp}(\varepsilon_1 2^{\varepsilon_2N}/N)\right)$  $\nearrow$& / \\
\hline
Analog & $\mathcal{O}\left(\frac{1}{P}\right)$ $\searrow$ &$\rightarrow$ $G_{\text{A}}^{\infty}$ & $\mathcal{O}\left(\frac{1}{N}\right)$ $\searrow$ & $\mathcal{O}\left(\frac{1}{\rho^2}\right)$ $\nearrow$  \\
\hline
\end{tabular}   
\end{center}  
\begin{tablenotes}
\item * The upward arrow indicates amplification at a certain order, while the downward arrow has the opposite meaning. The horizontal arrow indicates that it ultimately tends towards a fixed value.
\end{tablenotes}
\vspace{-0.2cm}
\end{table*}

\subsubsection{Impact of Transmit Power}
At low SNR levels, the achievable optimality gap under the digital transmission, $G_{\text{D}}$, vanishes as $\mathcal{O}\left(\mathrm{exp}\left({\varepsilon}/{P_{\max}}\right)\right)$ with the maximum transmit power budget $P_{\max}$, where $\varepsilon\triangleq \max_k \left \{ \frac{B N_0 \theta}{2N d_k^{-\alpha}}\right \}$. At high SNR regime, i.e., $P_{\max}\to \infty$, the successful transmission probability $p_k \to 1$, $\forall k$ and $G_{\text{D}}$ tends to
\begin{align}\label{infd}
G_{\text{D}}^{\infty}\triangleq \lim_{P_{\max}\to \infty} G_{\text{D}}=\frac{\eta (L \phi(b) +2L^3\delta^2)K}{2\mu N-4\eta L^2K}.
\end{align}
On the other hand, the decay rate for $G_{\text{A}}$ is equal to $\mathcal{O}\left({1}/{P_{\max}}\right)$ with low SNR values and the high SNR-limiting value is
\begin{align}\label{infa}
G_{\text{A}}^{\infty}\triangleq \lim_{P_{\max}\to \infty} G_{\text{A}}=\frac{2\eta L^3\delta^2\left(Kc-N\right)}{2\mu N-4\eta L^2\left(Kc-N\right)},
\end{align}
where $c\triangleq e^{\gamma_{\mathrm{th}}}+\frac{(1-\rho^2)\mathrm{E}_1\left(\gamma_{\mathrm{th}} \right)e^{2\gamma_{\mathrm{th}}}}{2\rho^2}$.

%\begin{remark}[Low SNR case]
%At low SNR regimes, the analog transmission achieves a lower optimality gap than the digital case. This is because with sufficiently small power budget, the successful transmission probability under the digital case tends to zero. In other words, no data can be effectively transmitted to the PS and hence the FL algorithm fails. On the contrary, the transmission is not completely truncated under the analog AirComp.
%\end{remark}

\begin{remark}
As SNR increases, the optimality gap for the analog case mainly comes from the non-IID datasets while the impact of the noise asymptotically diminishes. For the digital case, however, quantization errors additionally impose an impact. Under the analog transmission, the negative impact of non-IID datasets is enlarged due to imperfect AirComp. Imperfect CSI results in mismatched channel inversion in AirComp, rendering perfect computation of weighted sum impossible. 
Moreover, the performance degradation brought by imperfect CSI in the analog transmission cannot be mitigated by occupying more resources. Conversely, in the digital transmission, the convergence performance can be improved by occupying additional resources for increasing the number of quantization bits.
\end{remark}

\subsubsection{Impact of Device Number}
With the increasing number of participating devices, $N$, the virtual sum rate for the analog transmission, $g_{\text{A}}(\mathbf{r},\gamma_{\mathrm{th}})$, decreases at a rate of $\frac{1}{N}$, i.e., a faster convergence rate is achieved. As for the optimality gap, the impact of non-IID datasets asymptotically dominates $G_{\text{A}}$ and the decay rate is equal to $\mathcal{O}(1/N)$. Due to the involvement of more devices, a more accurate global gradient is obtained at the PS, which in turn facilitates the FL convergence and leads to better performance. Meanwhile, since different devices involved in the AirComp share the same time-frequency resource, an increase in access devices causes no deterioration of the AirComp performance, fully capturing the performance gain from more participating devices.

On the other hand, for the digital case, convergence performance does not necessarily monotonically change with $N$. Although more participating devices do bring performance gains,
it also leads to a significant deterioration of the transmission performance considering that limited communication resources are divided among additional users. Thus the convergence is compromised between communication reliability and the computation accuracy for wireless FL.
Specifically, the optimality gap, $G_{\text{D}}$, enlarges with a rate of $\mathcal{O}\left(\frac{1}{N}\mathrm{exp}(\varepsilon_1 2^{\varepsilon_2 N}/N)\right)$ with sufficiently large $N$, where $\varepsilon_1 =\frac{BN_0}{2Pd_K^{-\alpha}}$ and $\varepsilon_2=\frac{b+1}{M}$.

\begin{remark}
Benefiting from the characteristics of AirComp, more participating devices in the analog transmission always lead to performance improvement regardless of other parameters. Hence, allowing all active devices to participate in the FL training is the best choice for analog transmission. By contrast, in the digital transmission, it is necessary to seek a balance between the transmission performance and diversity gain through an optimization of $N$.
\end{remark}

\subsubsection{Impact of Imperfect CSI}
The imperfect CSI at the transmitter only affects the performance of analog transmission, which deteriorates at the order of $\frac{1}{\rho^2}$. Due to  imperfect CSI, the aggregation computation and the truncation decision in AirComp are contaminated, thus leading to a mismatch in the model aggregation and the impact of noise amplification. 

\begin{remark}
After incorporating computation capabilities into the analog case,  the emergence of computation error as a new source of error has positioned computational accuracy as a crucial factor affecting the convergence performance. It is concluded that CSI is a key factor affecting the performance gain brought by AirComp. Moreover, the truncation threshold $\gamma_{\mathrm{th}}$ should be optimized to adapt different levels of channel estimation accuracy. It can be effectively solved via bisection search in \cite{imperfect}.
\end{remark}

\subsubsection{Impact of the Number of Quantization Bits}
In the digital transmission, the number of quantization bits, $b$, also influences the FL performance in the following implicit ways. By selecting the minimum feasible $\theta = 2^{\frac{Nd(b+1)}{BT_{\max}}}-1$ in (\ref{eq22}), the achievable optimality gap $G_{\text{D}}$ is rewritten as
\begin{align}
    G_{\text{D}}
    &\approx \frac{\eta}{2\mu} \left(\frac{L\Delta^2}{(2^b-1)^2} +2L^3\delta^2\right) g_{\mathrm{D}}(\mathbf{r},b) \nonumber \\
    &= \frac{\eta}{2\mu} \left(\frac{L\Delta^2}{(2^b-1)^2} +2L^3\delta^2\right) \nonumber \\
    &\quad \times \left(\sum_{k=1}^K \frac{\alpha_k}{r_k}\mathrm{exp}\left(\frac{BN_0\left(2^{\frac{Nd(b+1)}{BT_{\max}}}-1\right)}{2NP_k d_k^{-\alpha}}\right)\right),
\end{align}
where the approximation is obtained in region of $\eta \ll \frac{\mu}{2L^2 g_{\text{D}}(\mathbf{r},b)}$. It is found that as $b$ increases, $G_{\text{D}}$ tends to first decrease and then increase. This is due to the diminishing quantization error term $\phi(b)$ with an increasing quantization accuracy and finally $G_{\text{D}}$ is dominated by the impact of packet loss. Therefore, it is necessary to optimize of the integer variable $b$ to pursue better convergence performance, which can be solved by a low-complexity exhaustive search method.

\subsection{Convergence Analysis under the Average Power Budget}
We consider the convergence with the average transmit power budget. For the digital transmission, by replacing $P_{\max}$ with $P_{\mathrm{ave}}$, we derive the similar results as \emph{Theorem 1} and is omitted here due to page limit. As for the analog transmission, we have the following corollary.

\begin{corollary}
For a fixed learning rate satisfying $\eta \leq \frac{\mu}{2L^2 g_{\text{A}}(\mathbf{r},\gamma_{\mathrm{th}})}$, the optimality gap of the distributed gradient update in the $(m+1)$-th iteration under the analog transmission follows
\begin{align}
\mathbb{E}&\left[ F(\tilde{\mathbf{w}}_{m+1})\right]-F(\mathbf{w}^*)\nonumber \\
 &\leq \frac{L}{2}\left( 1-\eta \mu +2\eta^2 L^2g_{\text{A}}(\mathbf{r},\gamma_{\mathrm{th}})\right)^{m+1} \mathbb{E}\left[\left \Vert \tilde{\mathbf{w}}_{0} -\mathbf{w}^*\right \Vert^2 \right]\nonumber \\
&\quad +\frac{\eta \left(L\varphi_{\mathrm{ave}}(\mathbf{r},\gamma_{\mathrm{th}})+2L^3 \delta^2 g_{\text{A}}(\mathbf{r},\gamma_{\mathrm{th}}) \right)}{2\mu-4\eta L^2 g_{\text{A}}(\mathbf{r},\gamma_{\mathrm{th}})},
\end{align}
where $\varphi_{\mathrm{ave}}(\mathbf{r},\gamma_{\mathrm{th}})\triangleq \frac{B N_0 \gamma^2 e^{2\gamma_{\mathrm{th}}}\mathrm{E}_1(\gamma_{\mathrm{th}})}{2P_{\mathrm{ave}}\rho^2 }\max_k\left \{\frac{\alpha_k^2}{r_k^2 }d_k^{\alpha}\right\}$. The optimality gap with sufficient iterations follows
\begin{align}\label{Ga2}
G_{\text{A},\mathrm{ave}}= \frac{\eta \left(L\varphi_{\mathrm{ave}}(\mathbf{r},\gamma_{\mathrm{th}})+2L^3 \delta^2 g_{\text{A}}(\mathbf{r},\gamma_{\mathrm{th}}) \right)}{2\mu-4\eta L^2 g_{\text{A}}(\mathbf{r},\gamma_{\mathrm{th}})}.
\end{align}
\end{corollary}

\begin{proof}
please refer to Appendix \ref{coll1}. \hfill $\square$
\end{proof}

\begin{remark}
It is worth noting that $\mathrm{E}_1(\gamma_{\mathrm{th}})<\frac{1}{\gamma_{\mathrm{th}}}$ when $\gamma_{\mathrm{th}}>0$. Compared with the maximum transmit power budget, a smaller optimality gap for the analog transmission is achieved with the average power budget. Due to the need for channel alignment in AirComp, the performance is dominantly limited by the device with the worst channel condition. Furthermore, the strict peak power constraint amplifies the impact of worst-case channel conditions, resulting in looser convergence performance compared to the long-term constraint. 
%However, unlike the analog case, different devices do not interfere with each other under the FDMA based digital transmission, thus transmitting at full power can achieve the optimal convergence.
\end{remark}

To summarize, while the analog AirComp improves the spectrum utilization compared to the digital  paradigm, it faces challenges in fully utilizing the power resource, particularly with strict peak power constraints. Conversely, orthogonal access in digital transmission is not suitable for scenarios with massive access due to the limitations in spectrum resources.
%In addition, the heterogeneity of wireless channels is inevitably introduced in the digital transmission, which needs to be carefully suppressed through the device sampling optimization discussed in the following section.

\subsection{Discussions on Scenarios with Advanced System Designs}
To facilitate  performance analysis, we introduce assumptions regarding the system design, including multiple access, parameter quantization, and power control methods. Subsequently, we delve into the implications of advanced system designs on the FL performance and comparison.

In the digital transmission, the FL performance can primarily be improved from two aspects, namely enhancing transmission reliability and optimizing resource utilization. Specifically, advanced transmissions strategies help minimize transmission errors and packet losses due to channel fading. Furthermore, if other resource allocation methods, such as the model compression design and device scheduling strategies, are exploited toward the FL tasks, they prioritize crucial parameter/device transmissions and thus lifting the resource utilization. On the other hand, in the analog transmission, the FL performance through AirComp is primarily influenced by the over-the-air computational accuracy. Optimized transceiver and power control designs help mitigate the negative impact of channel fading on the FL performance.

While further optimization of system designs enhances performance, it is essential to note that the performance limits for the digital and analog transmissions remains unchanged. As observed in the above analytical results, in the digital transmission paradigm, due to the decoupling of the communication and computation processes, the number of bits that can be accurately transmitted with the limited resources is determined, which places an upper bound of the FL performance. In contrast, within the analog transmissions, the receiver does not aim to recover information from individual sources but instead prioritizes the precision of computation results derived from the over-the-air superimposed signals, thereby making computational accuracy a decisive role. Hence, the performance limit of the analog transmission is contingent upon the channel estimation accuracy and additive noise level.

\section{Device Sampling Optimization}
Based on the derived results in Section IV, we are able to further establish an optimization design of the device sampling for the wireless FL to improve the convergence.

\subsection{Digital Transmission}
By direct inspection of (\ref{Gd}), the optimality gap $G_{\text{D}}$ monotonically decreases with a decreasing virtual sum weight. Hence, the device sampling optimization problem with the digital transmission is formulated as
\begin{align}\label{p11}
\mathop{\text{minimize}}_{\mathbf{r}} &\quad g_{\text{D}}(\mathbf{r},b)=\sum_{k=1}^K \frac{\alpha_k}{p_k r_k}\nonumber \\
\text{subject to}& \quad \sum_{k=1}^K r_k =N,\enspace  r_k\leq 1,\enspace k=1,2,\cdots,K,
\end{align}  
which is a convex problem. By exploiting the KKT conditions, we obtain the optimal inclusion probability as 
\begin{align}\label{optr}
r_k^*=\min \left\{ \sqrt{\frac{\alpha_k}{\nu p_k}},\enspace 1\right\},
\end{align}
where $\nu$ is the Lagrangian multiplier and it is selected to satisfy $\sum_{k=1}^K r_k^*=N$. Note that the value of $\sum_{k=1}^K r_k^*$ varies monotonically with $\nu$ and thus we can rely on a bisection-based search method \cite{energy} to get the optimal solution of problem (\ref{p11}).

\begin{remark}
The optimal inclusion probability is positively correlated with the local dataset size while it behaves conversely correlated with the successful transmission probability. In other words, a device with a larger dataset is deemed more important for model training, thereby deserving a sampling bias. Conversely, devices with lower successful transmission probabilities contribute less to the model training process, requiring more frequent sampling to compensate. Thus, the goal of our inclusion probability optimization is to address the imbalances in the dataset size, and the heterogeneity introduced by uneven channel fading. It ensures fair and effective participation among diverse devices. 
\end{remark}

Moreover, note that the influence of quantization error and data heterogeneity are equally amplified by $g_\text{D}(\mathbf{r},b)$. It indicates that the optimization of inclusion probabilities $\mathbf{r}$ cannot adequately adapt to varying local data distributions.

\begin{figure*}[!t]
\subfigure[]{
\begin{minipage}[t]{0.5\linewidth}
  \centering
  \includegraphics[width=3in]{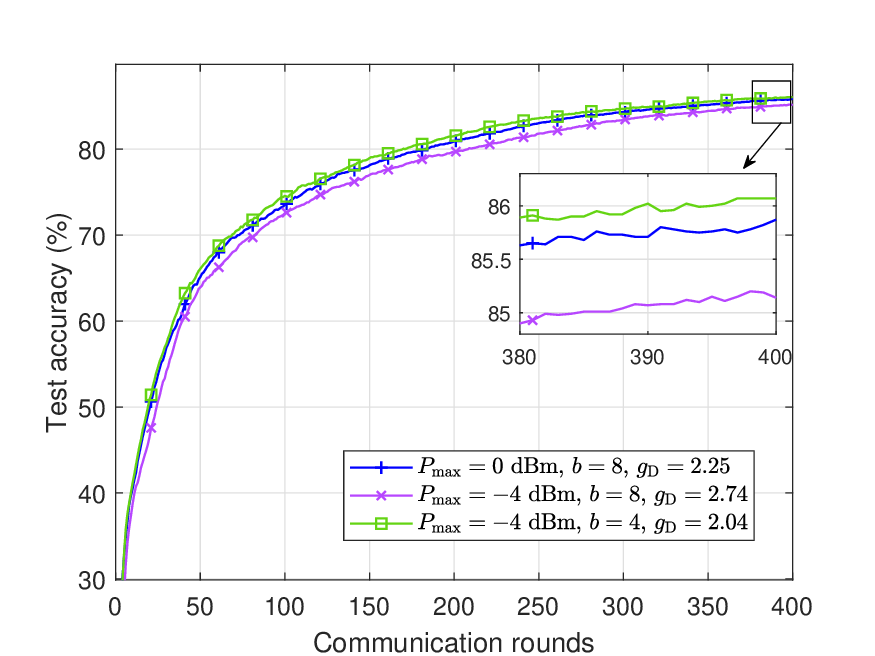}
 \end{minipage}
}
\subfigure[]{
\begin{minipage}[t]{0.5\linewidth}
  \centering
  \includegraphics[width=3in]{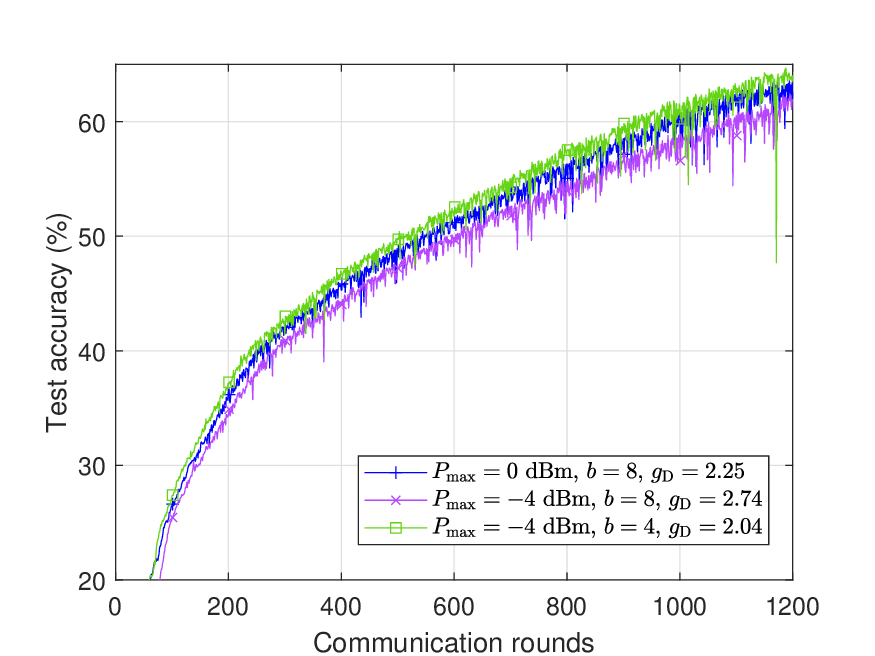}
 \end{minipage}
}
\caption{Convergence performance under digital transmission: (a) MNIST dataset, (b) CIFAR-10 dataset\label{fig3}.}
\vspace{-0.5cm}
\end{figure*}

\begin{figure*}[!t]
\subfigure[]{
\begin{minipage}[t]{0.5\linewidth}
  \centering
  \includegraphics[width=3in]{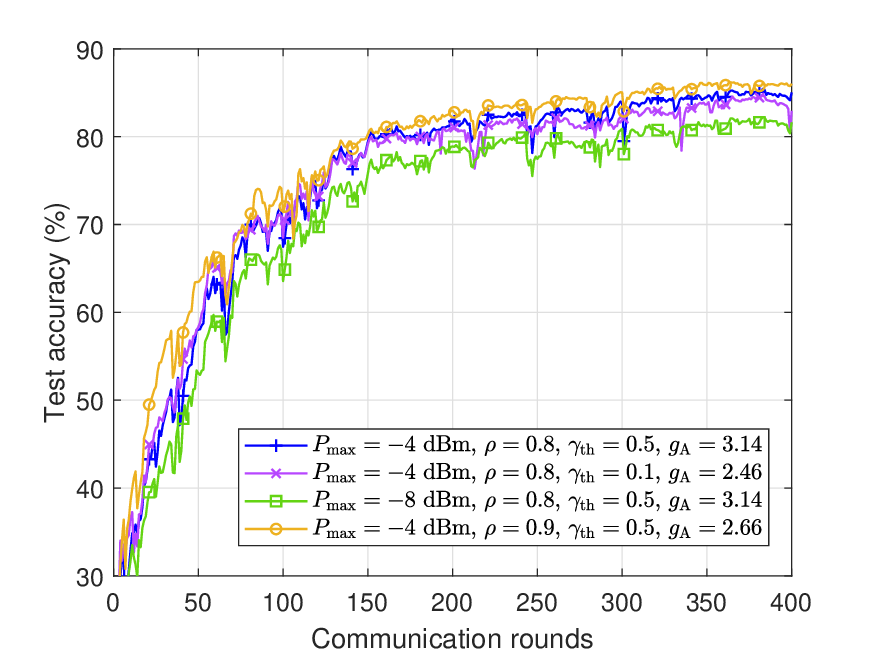}
 \end{minipage}
}
\subfigure[]{
\begin{minipage}[t]{0.5\linewidth}
  \centering
  \includegraphics[width=3in]{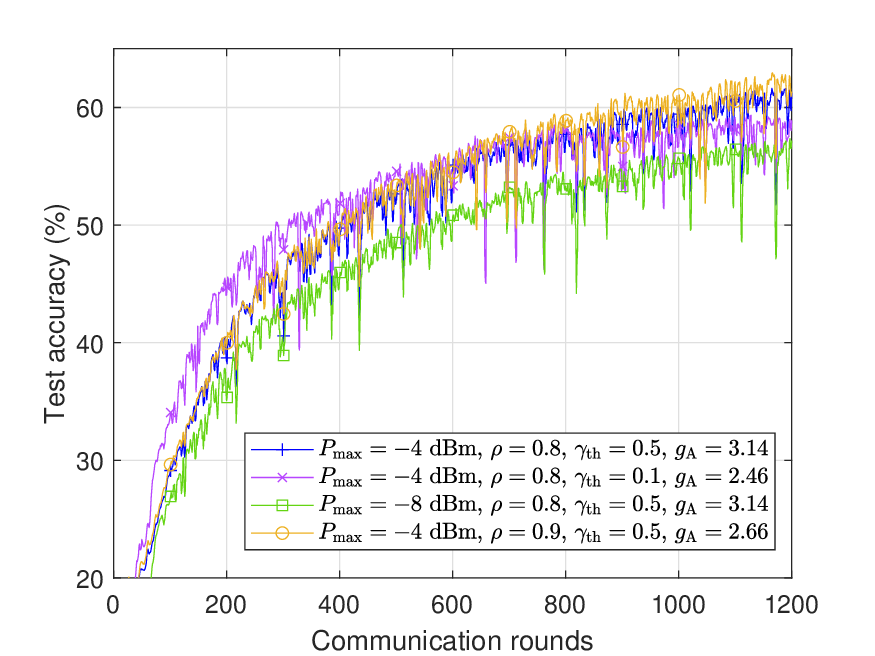}
 \end{minipage}
}
\caption{Convergence performance under analog transmission: (a) MNIST dataset, (b) CIFAR-10 dataset.\label{fig32}}
\vspace{-0.5cm}
\end{figure*}

\subsection{Analog Transmission}
As for the analog transmission, the device sampling optimization is expressed as
\begin{align}\label{p21}
\mathop{\text{minimize}}_{\mathbf{r}} &\quad \frac{\varphi(\mathbf{r},\gamma_{\mathrm{th}})+2L^2 \delta^2 g_{\text{A}}(\mathbf{r},\gamma_{\mathrm{th}}) }{2\mu-4\eta L^2 g_{\text{A}}(\mathbf{r},\gamma_{\mathrm{th}})}\nonumber \\
\text{subject to}&\quad  \sum_{k=1}^K r_k =N,\enspace  r_k\leq 1,\enspace k=1,2,\cdots,K.
\end{align}
Note that under the average transmit power budget, (\ref{Ga2}) only differs from the objective value in the constant term, and hence we will not discuss it separately. Considering the intractable fractional form of the objective function in (\ref{p21}), we rely on the well-known Dinkelbach algorithm for reformulation \cite{dinkel,fp2}. According to the definition of $\varphi(\mathbf{r},\gamma_{\mathrm{th}})$ and $g_{\text{A}}(\mathbf{r},\gamma_{\mathrm{th}})$ in (\ref{Ga}), it is easy to check that the denominator of the objective function in (\ref{p21}) is concave and the numerator is convex. Hence, the iterative Dinkelbach algorithm guarantees to converge to the global optimum of (\ref{p21}). Concretely, in the $t$-th iteration, we reformulate the problem in (\ref{p21}) as 
\begin{align}\label{p22}
\mathop{\text{minimize}}_{\mathbf{r}}  &\quad \varphi(\mathbf{r},\gamma_{\mathrm{th}})+(2L^2 \delta^2 +4\eta L^2 \varsigma^{(t-1)})g_{\text{A}}(\mathbf{r},\gamma_{\mathrm{th}}) \nonumber \\
\text{subject to}&\quad  \sum_{k=1}^K r_k =N,\enspace  r_k\leq 1,\enspace k=1,2,\cdots,K.
\end{align}
where $\varsigma^{(t-1)}$ is a constant determined in the previous round. Note that the problem in (\ref{p22}) is convex and thus can be solved by numerical convex program solvers, e.g., CVX tools \cite{cvx}. After obtaining the optimal $\mathbf{r}^{(t)}$ of the $t$-th subproblem in (\ref{p22}),  the  auxiliary constant is updated as
\begin{align}\label{update}
\varsigma^{(t)}=\frac{\varphi(\mathbf{r}^{(t)},\gamma_{\mathrm{th}})+2L^2 \delta^2 g_{\text{A}}(\mathbf{r}^{(t)},\gamma_{\mathrm{th}}) }{2\mu-4\eta L^2 g_{\text{A}}(\mathbf{r}^{(t)},\gamma_{\mathrm{th}})}.
\end{align}
Iterating the above steps until convergence, we obtain the optimal $\mathbf{r}$ of the problem in (\ref{p21}).

\begin{remark}
Unlike the digital transmission case, the device sampling optimization is committed to seeking a trade-off between the equivalent noise power $\varphi(\mathbf{r},\gamma_{\mathrm{th}})$ and virtual sum weight $g_{\text{A}}(\mathbf{r},\gamma_{\mathrm{th}})$, and the parameter $\delta$ functions as a weighting factor to facilitate the optimal trade-off. At high SNR regimes or with extremely uneven local data distributions, the noise term is comparably ignorable and hence the optimality gap is dominated by $g_\text{A}$. Hence, the optimization of $\mathbf{r}$ is isolated from specific channel conditions and only needs to match the size of local datasets.
\end{remark}

\section{Numerical Results}

In this section, we provide simulation results to verify the performance analysis and the inclusion probability optimization. We deploy $K=20$ edge devices uniformly distributed  in a square area with radius $500$ m and a PS at the center of the square area. The most popular MNIST dataset and CIFAR-10 dataset are exploited for the FL performance evaluation.
The MNIST dataset contains 10 classes of handwritten digits ranging from 0 to 9 and we train a multi-layer perceptron (MLP) with $d=23,860$ parameters via the wireless FL algorithm for classification purposes. Moreover, the CIFAR-10 dataset includes 10 classes with labels 0-9 and we train a convolutional neural network (CNN) with $d=60,000$ parameters. The trained CNN contains two convolutional layers and three fully connected layers. Max pooling operation is conducted following each convolutional layer and the activation function is ReLU. Different edge devices own different data samples, and each local dataset has up to two types of data samples to capture the non-IID characteristic.

Unless otherwise specified, the other parameters are set as: the number of participating devices $N=10$, the bandwidth, $B=1$ MHz, the path loss exponent, $\alpha=3$,  the noise power $N_0=-80$ dBm/Hz, the maximum transmit power budget, $P_{\max}=0$ dB, the number of quantization bits, $b=8$, the truncation threshold, $\gamma_{\mathrm{th}}=0.5$, the delay target $T_{\max}$ is equal to $T_{\text{A}}$ in (\ref{delaya}), and the learning rate $\eta=0.01$. We set $L=8$ and $\mu=2$, which fall within the existing typical range of values in \cite{value1,value2}. Additionally, the parameter $\delta$, serving as an upper bound of $\left \Vert \mathbf{w}_k^*-\mathbf{w}^* \right \Vert^2$, is estimated through simulation tests. 

\subsection{Convergence Performance}

\begin{figure}[!t]
\centering
\includegraphics[width=3in]{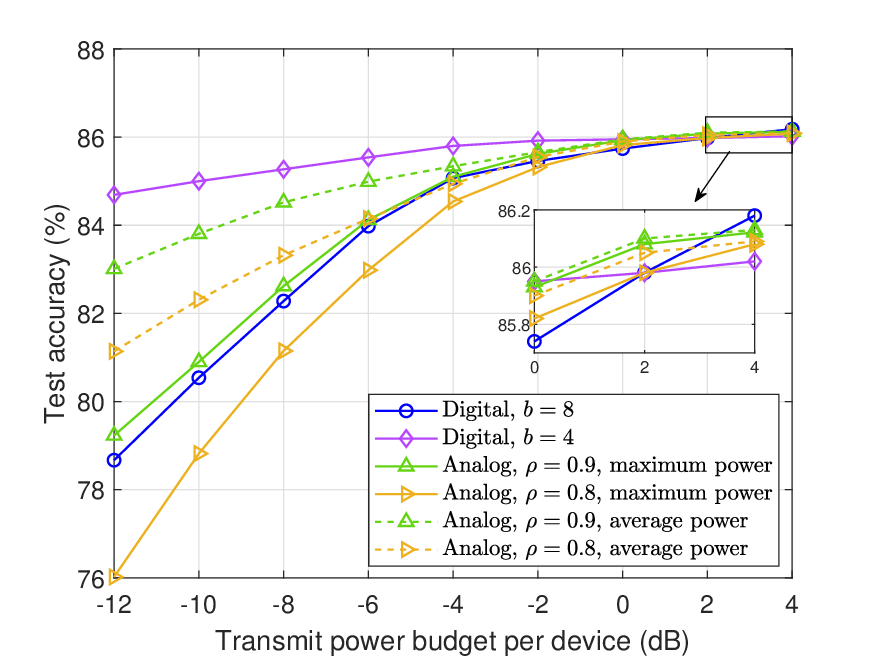}
\caption{ Test accuracy versus transmit power budget.}\label{fig4}
\vspace{-0.5cm}
\end{figure}
\begin{figure}[!t]
\centering
\includegraphics[width=3in]{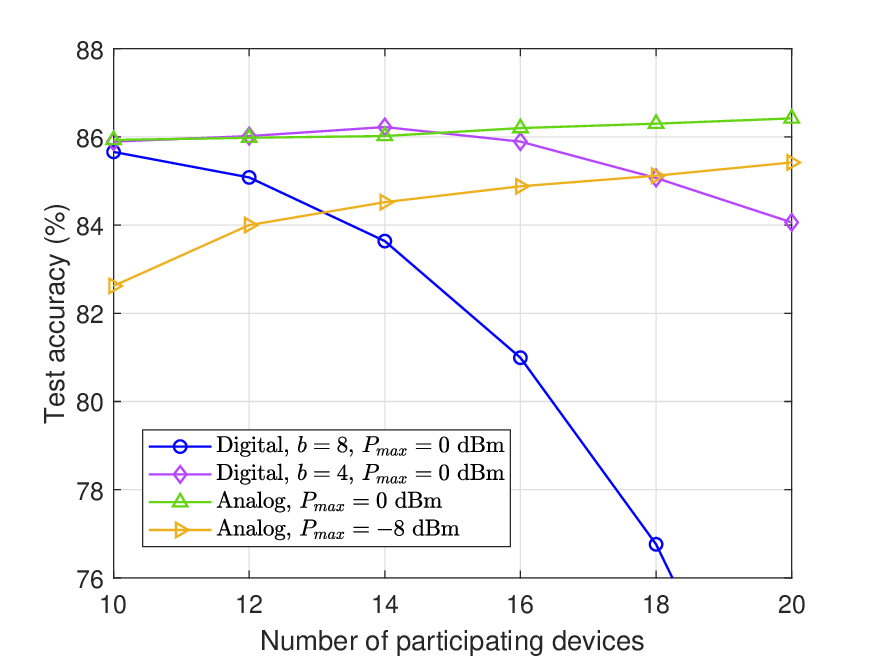}
\caption{ Test accuracy versus the number of participating devices.}\label{fig5}
\vspace{-0.5cm}
\end{figure}

In Figs. \ref{fig3} and \ref{fig32}, we depict the convergence performance for the digital and analog transmission. As shown in Fig. \ref{fig3}, we observe that the convergence rate and optimality gap under digital transmission exhibit a negative correlation with the virtual sum weight, aligning with our theoretical analysis. Moreover, the convergence behavior remains consistent with the analytical results despite the complexity of the classification task, thereby validating the accuracy of the theoretical analysis.

\begin{figure}[!t]
\centering
\includegraphics[width=3in]{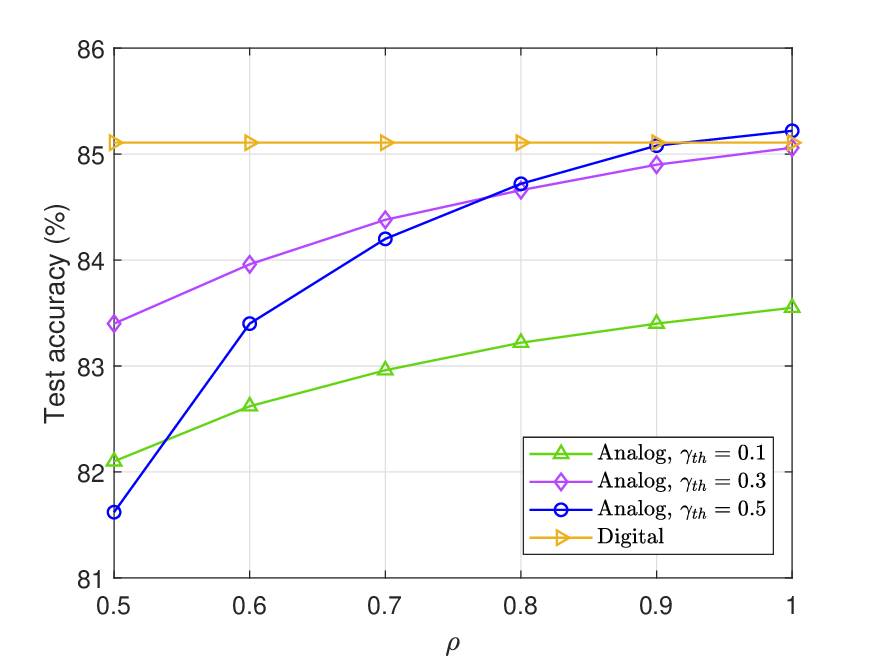}
\caption{ Test accuracy versus the accuracy level of channel estimation accuracy.}\label{fig6}
\vspace{-0.5cm}
\end{figure}

For the analog case depicted in Fig. \ref{fig32}, consistent with the analytical findings, we notice that the convergence rate is negatively correlated with the virtual sum weight $g_\text{A}$, which is determined by $\rho$ and $\gamma_{\mathrm{th}}$. On the other hand, transmit power only affects the achievable optimality gap after convergence. This is because changes in transmit power only affect the equivalent power of the additive noise. Additionally, modifications in $\rho$ and $\gamma_{\mathrm{th}}$ affect the distortion of the aggregation coefficient, which in turn influences the computation error. Furthermore, the increased complexity of FL tasks renders fluctuations in the performance curve more sensitive to noise. Consequently, in the analog communication, the superimposed white Gaussian  noise is significantly severer than quantization errors observed in the digital transmission, thus leading to more pronounced fluctuations in convergence performance. It implies that for more complex learning tasks, it becomes imperative to further reduce the variance of gradient estimation to mitigate excessive fluctuations and their adverse impacts on convergence.

\subsection{Impact of Transmit Power Budget}

In Fig. \ref{fig4}, we show the test accuracy versus different transmit power budgets. It is observed that the digital transmission scheme outperforms the analog scheme, particularly with high SNR levels. In such cases, employing more quantization bits yields the best performance. Conversely, for low SNR levels, reducing the quantization bits leads to marginal performance loss, highlighting the flexibility of the digital schemes by selecting different quantization accuracies. On the other hand, the analog scheme faces significant performance limitations, particularly with the maximum transmit power budget and less CSI, due to the stringent requirements of channel inversion. Therefore, in terms of power utilization, the digital scheme is more efficient than the analog counterpart. 

\begin{figure*}[!t]
\subfigure[]{
\begin{minipage}[t]{0.5\linewidth}
  \centering
  \includegraphics[width=3in]{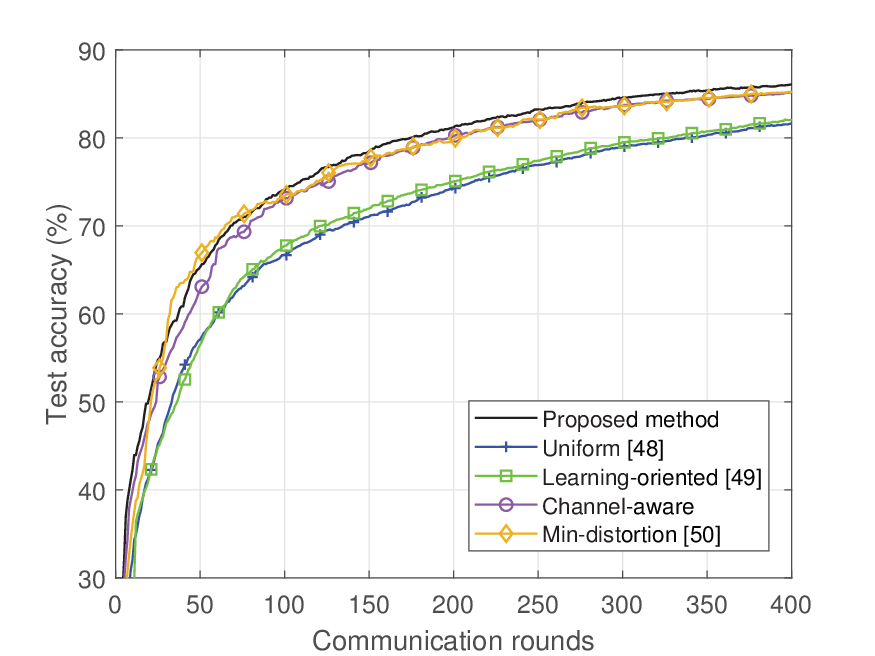}
 \end{minipage}
}
\subfigure[]{
\begin{minipage}[t]{0.5\linewidth}
  \centering
  \includegraphics[width=3in]{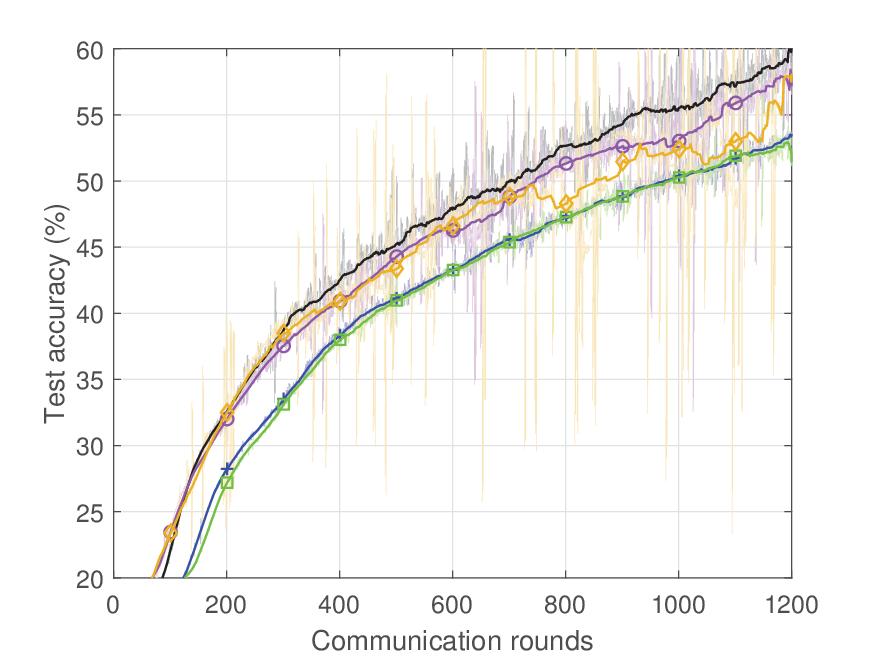}
 \end{minipage}
}
\caption{Convergence performance with different inclusion probabilities and digital transmission: (a) MNIST dataset, (b) CIFAR-10 dataset\label{fig7}.}
\vspace{-0.5cm}
\end{figure*}

\begin{figure*}[!t]
\subfigure[]{
\begin{minipage}[t]{0.5\linewidth}
  \centering
  \includegraphics[width=3in]{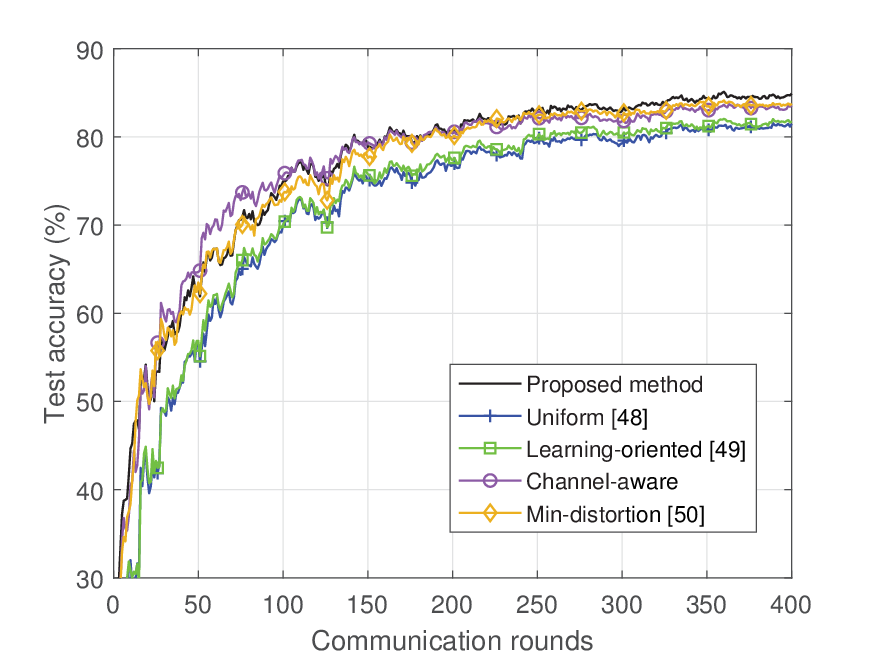}
 \end{minipage}
}
\subfigure[]{
\begin{minipage}[t]{0.5\linewidth}
  \centering
  \includegraphics[width=3in]{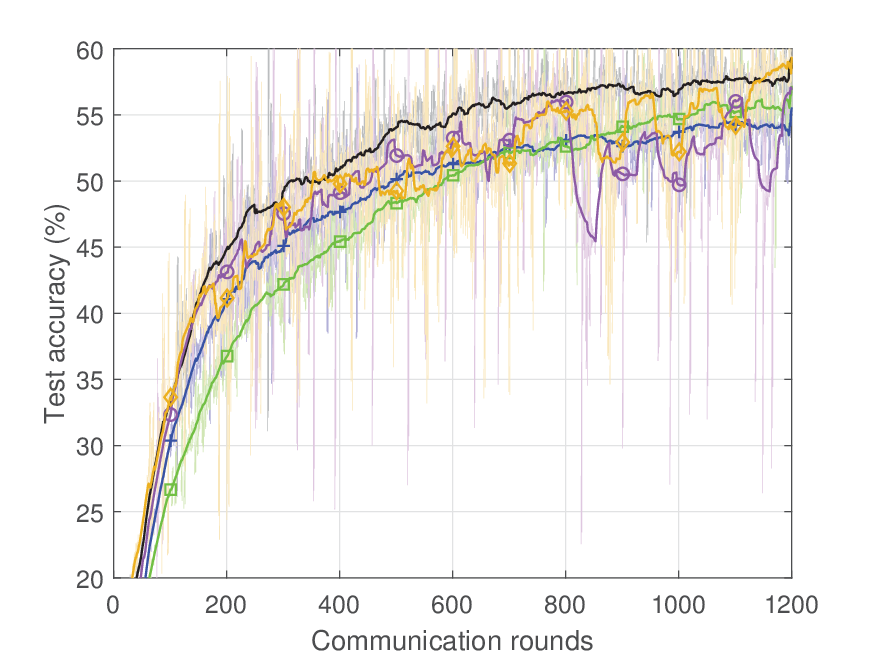}
 \end{minipage}
}
\caption{Convergence performance with different inclusion probabilities and analog transmission: (a) MNIST dataset, (b) CIFAR-10 dataset.\label{fig8}}
\vspace{-0.5cm}
\end{figure*}

\subsection{Impact of Participating Device Numbers}

Fig. \ref{fig5} illustrates the test accuracy versus the number of participating devices. We note that for the analog transmission, the test accuracy gradually increases as $N$ increases. In contrast, although the performance in digital case may be improved initially, it eventually decline rapidly as each device can only occupy a limited amount of resources, making it unable to support high-rate transmission. Consequently, the results suggest that for digital transmission, the selection of $N$ requires further optimization according to the actual conditions, with a preference for fewer devices.

\subsection{Impact of Channel Estimation Accuracy}

In Fig. \ref{fig6}, we present the impact of channel estimation accuracy on the analog case. It is evident that better performance can be achieved with more accurate CSI. 
%This is because in AirComp, the integration of communication and computation makes CSI crucial as it affects both the SNR at the receiver and the occurrence of computational errors. 
Additionally, we observe that smaller truncation thresholds are more suitable for larger $\rho$, while larger truncation thresholds are preferred for smaller $\rho$. This is because higher CSI uncertainties have a significant impact on truncation choices, necessitating looser truncation conditions to reduce incorrect choices.

\subsection{Impact of Different Inclusion Probabilities}

In Figs. \ref{fig7} and \ref{fig8}, we depict the convergence performance with different inclusion probabilities. 
For comparison, we consider the following baselines for comparison. For the sake of fairness, all schemes refrain from utilizing specific information on instantaneous CSI and gradients.
\begin{itemize}
    \item Uniform \cite{uniform}: The inclusion probabilities are uniformly assigned the same value, i.e., $p_k =\frac{1}{K}$.
    \item Learning-oriented \cite{baseline2}: From the perspective of learning algorithms, the probability is set to be proportional to the size of the local datasets, i.e.,  $p_k \propto \alpha_k$.
    \item Channel-aware: From the perspective of wireless channels, the probability is set to be proportional to the large-scale path loss, i.e.,  $p_k \propto d_k^{-\frac{\alpha}{2}}$. 
    \item Min-distortion \cite{baseline3}: To minimize the communication distortion in the analog transmission, the probability is set to be proportional to $\alpha_k d_k^{\frac{\alpha}{2}}$ by considering both the local datasets and channel conditions.
\end{itemize}

As shown in Fig. \ref{fig7}, the proposed method consistently outperforms the aforementioned baseline methods across all levels. The first two baselines neglect the influence of the wireless transmission process, resulting in performance degradation. The sampling method based on channel conditions tends to select devices with better channels, effectively reducing packet loss rates and yielding significant performance improvements. However, due to its oversight of imbalanced size of local datasets, its final performance remains inferior to our proposed method. The fourth baseline, tailored for the analog transmission scenarios, partially accounts for the impact of local datasets and wireless channels but lacks optimality, leading to limited performance gains.

As for the analog transmission case in Fig. \ref{fig8}, we note that although the performance of the optimized probability is superior, the performance gain compared to the other baselines is not significant. This limit arises from the reliance on constants $L$, $\mu$, and $\delta$ in the optimization problem (\ref{p21}), which are challenging to determine accurately in practice, thus affecting the final performance. Similarly, akin to the digital transmission, the sampling method based on channel conditions effectively mitigates the negative impact of the imperfect wireless transmission. However, its disregard for data characteristics results in suboptimal performance, particularly in the complex classification tasks on CIFAR-10 dataset, leading to significant performance fluctuations. Furthermore, the baseline method of minimizing computational distortion overlooks the impact of data heterogeneity, thus impeding its ability to achieve satisfactory performance.

\section{Conclusion}

In this paper, we have provided a detailed comparison between digital and analog transmission enabled wireless FL. To this end, we considered general transmission designs for both schemes and conducted a fair comparison between them. Then, we analyzed the convergence behavior of wireless FL in terms of the convergence rate and optimality gap under digital and analog cases, and compared the convergence performance from multiple perspectives. It was found that  digital transmission is more suitable for scenarios with sufficient radio resources and CSI uncertainties. On the other hand, analog transmission is suitable when their are massive numbers of participating devices. Next, we addressed sampling optimization for both cases, and further developed insights for optimization, which ars useful for practical deployment. Finally, experimental results illuminated the analytical results and the sampling strategies. Additionally, an explicit and precise characterization of data heterogeneity and targeted system designs with theoretical guarantees should be of our interest in the future work.

\appendices
\section{Proof of Lemma 2}\label{lemma2}
For the digital case, according to \cite[\emph{Lemma 5}]{quan1}, we first conclude that the quantized gradients $\mathcal{Q}(\mathbf{g}_m^k)$ is unbiased, i.e., 
\begin{align}
\mathbb{E}\left [ \mathcal{Q}(\mathbf{g}_m^k)\right ]=\mathbf{g}_m^k.
\end{align}
Combining with the fact that $\mathbb{E}\left[\xi_{k,\text{D}}\right]=1$ in (\ref{distortiond}), we have
\begin{align}\label{eq47}
    \mathbb{E}\left[\hat{\mathbf{g}}_{m,\text{D}}\right]&\overset{\text{(a)}}{=} \sum_{k=1}^K \alpha_k
    \mathbb{E}\left[\frac{\chi_k}{r_k}\right]\mathbb{E}\left[\xi_{k,\text{D}}\right]\mathbb{E}\left [ \mathcal{Q}(\mathbf{g}_m^k)\right ]\nonumber \\
    &=\sum_{k=1}^K \alpha_k \mathbf{g}_m^k=\mathbf{g}_m,
\end{align}
where (a) comes from the definition of $\hat{\mathbf{g}}_{m,\text{D}}$ and the independence among device sampling, small-scale fadings and stochastic quantization.

As for the analog transmission, by exploiting \cite[\emph{Lemma~1}]{imperfect}, we have $\mathbb{E}\left[ \xi_{k,\text{A}}\right]=1$. Combining with the statistical characteristic of $\chi_k$ and $\bar{\mathbf{z}}_m$ and following the same procedures in (\ref{eq47}), we get the desired conclusion, i.e., $\mathbb{E}\left[\hat{\mathbf{g}}_{m,\text{A}}\right]=\mathbf{g}_m$. The proof completes.

%and the quantization error follows
%\begin{align}
%\mathbb{E}\left [\left \Vert\mathcal{Q}(\mathbf{g}_m^k)-\mathbf{g}_m^k\right \Vert^2 \right ] &\leq \frac{d}{4}\left(\frac{ g_{m,\max}^k-g_{m,\min}^k}{2^b-1}\right)^2 \nonumber \\
%&\leq \frac{\Delta^2}{(2^b-1)^2}\triangleq \phi(b)
%\end{align}
%where $\Delta^2$ is a uniform upper bound of $\frac{d}{4}\left(g_{m,\max}^k-g_{m,\min}^k \right)^2$, $\forall m,k$.

%Starting from Lemma 1, we can easily prove that $\hat{\mathbf{g}}_{m,\text{D}}$ is also an unbiased estimate of the actual gradient in (\ref{gradient}). Moreover, the variance of quantization error is directly related to the number of quantization bits and can be completely eliminated with sufficient large $b$.

\section{Proof of Theorem 1}\label{theo1}
To begin with, we define an auxiliary variable as
\begin{align}
\hat{\mathbf{w}}_{m+1} =\tilde{\mathbf{w}}_m -\eta \mathbf{g}_m,
\end{align}
which represents the model obtained at $(m+1)$-th round via ideal communication and full participation. Then, by exploiting \emph{Assumption~2} and the fact that $\nabla F(\mathbf{w}^*)=0$, we have
\begin{align}\label{ap1}
&\mathbb{E}\left[ F(\tilde{\mathbf{w}}_{m+1})\right]-F(\mathbf{w}^*)\leq \frac{L}{2} \mathbb{E} \left [ \left\Vert \tilde{\mathbf{w}}_{m+1}-\mathbf{w}^* \right \Vert^2\right]\nonumber\\
&\quad\overset{\text{(a)}}{=}\frac{L}{2}\left(\underbrace{ \mathbb{E} \left [ \left\Vert \tilde{\mathbf{w}}_{m+1}-\hat{\mathbf{w}}_{m+1} \right \Vert^2\right]}_{A_1}+ \underbrace{\mathbb{E} \left [ \left\Vert \hat{\mathbf{w}}_{m+1}-\mathbf{w}^* \right \Vert^2\right]}_{A_2}\right),
\end{align}
where (a) is due to the fact that $\hat{\mathbf{g}}_{m,\text{D}}$ is an unbiased estimate of $\mathbf{g}_m$. 
For the term $A_1$, it is bounded by
\begin{align}\label{a1}
A_1&=\eta^2 \mathbb{E}\left [\left \Vert \hat{\mathbf{g}}_{m,\text{D}} -\mathbf{g}_m \right \Vert^2\right]\nonumber \\
&=\eta^2 \mathbb{E}\left [ \left \Vert \sum_{k=1}^K \frac{\chi_k \alpha_k\xi_{k,\text{D}}}{r_k} \mathcal{Q}(\mathbf{g}_m^k) -\sum_{k=1}^K \alpha_k \mathbf{g}_m^k \right \Vert^2\right]\nonumber \\
&\overset{\text{(a)}}{=}\eta^2 \mathbb{E}\left[\left \Vert \sum_{k=1}^K \alpha_k \left ( \frac{\chi_k\xi_{k,\text{D}}}{ r_k} \mathcal{Q}(\mathbf{g}_m^k) -\sum_{i=1}^K \alpha_i \mathbf{g}_m^i \right )\right \Vert^2 \right ]\nonumber \\
&\overset{\text{(b)}}{\leq}  \eta^2\sum_{k=1}^K \alpha_k \mathbb{E}\left[\left \Vert \frac{\chi_k\xi_{k,\text{D}}}{ r_k} \mathcal{Q}(\mathbf{g}_m^k) -\sum_{i=1}^K \alpha_i\mathbf{g}_m^i \right \Vert^2 \right ]\nonumber \\
&= \eta^2\sum_{k=1}^K \alpha_k \mathbb{E}\left[\left \Vert \left(\frac{\chi_k\xi_{k,\text{D}}}{ r_k} \mathcal{Q}(\mathbf{g}_m^k)-\mathbf{g}_{m}^k\right )\right.\right.\nonumber \\
&\quad\left.\left.+\left( \mathbf{g}_{m}^k-\sum_{i=1}^K \alpha_i\mathbf{g}_m^i \right)\right \Vert^2 \right ]\nonumber \\
&\overset{\text{(c)}}{=}\eta^2\underbrace{ \sum_{k=1}^K \alpha_k \mathbb{E}\left[\left \Vert \frac{\chi_k\xi_{k,\text{D}}}{ r_k} \mathcal{Q}(\mathbf{g}_m^k)-\mathbf{g}_{m}^k\right \Vert^2 \right ]}_{B_1}\nonumber \\
&\quad+\eta^2\underbrace{ \sum_{k=1}^K \alpha_k \mathbb{E}\left[\left \Vert \mathbf{g}_{m}^k-\sum_{i=1}^K \alpha_i\mathbf{g}_m^i \right \Vert^2 \right ]}_{B_2},
\end{align}
where (a) is because $\sum_{k=1}^K \alpha_k =1$, (b) exploits the convexity of $\Vert \cdot \Vert^2$, and (c) is due to the fact that $\mathbb{E} \left [ \frac{\chi_k\xi_{k,\text{D}}}{ r_k} \mathcal{Q}(\mathbf{g}_m^k)\right ]=\mathbf{g}_{m}^k$. According to \cite{quan2}, the variance of quantization error is bounded as
\begin{align}
\mathbb{E}\left [\left \Vert\mathcal{Q}(\mathbf{g}_m^k)-\mathbf{g}_m^k\right \Vert^2 \right ] &\leq \frac{d}{4}\left(\frac{ g_{m,\max}^k-g_{m,\min}^k}{2^b-1}\right)^2 \nonumber \\
&\leq \frac{\Delta^2}{(2^b-1)^2}\triangleq \phi(b)
\end{align}
where $\Delta^2$ is defined as a uniform upper bound of $\frac{d}{4}\left(g_{m,\max}^k-g_{m,\min}^k \right)^2$, $\forall m,k$. Then, $B_1$ is bounded by
\begin{align}\label{b1}
B_1&= \sum_{k=1}^K \alpha_k \mathbb{E}\left[ \left \Vert \left(\frac{\chi_k\xi_{k,\text{D}}}{r_k} \mathcal{Q}(\mathbf{g}_m^k)-\frac{\chi_k\xi_{k,\text{D}}}{ r_k} \mathbf{g}_m^k\right)\right. \right. \nonumber \\
&\quad \left. \left. +\left(\frac{\chi_k\xi_{k,\text{D}}}{ r_k} \mathbf{g}_m^k -\mathbf{g}_m^k\right)\right \Vert^2 \right ]\nonumber \\
&=\sum_{k=1}^K \alpha_k \mathbb{E}\left[ \left(\frac{\chi_k\xi_{k,\text{D}}}{r_k}\right)^2\right]\mathbb{E}\left[ \left \Vert \mathcal{Q}(\mathbf{g}_m^k) -\mathbf{g}_m^k \right\Vert^2\right]\nonumber \\
&\quad +\sum_{k=1}^K \alpha_k \mathbb{E}\left[ \left(\frac{\chi_k\xi_{k,\text{D}}}{r_k}-1\right)^2\right]\mathbb{E}\left[ \left \Vert \mathbf{g}_m^k \right\Vert^2\right]
\nonumber \\
&\overset{\text{(a)}}{\leq}\sum_{k=1}^K \frac{\phi(b) \alpha_k}{p_k r_k}+\sum_{k=1}^K \alpha_k \left( \frac{1}{p_k r_k}-1\right)\mathbb{E}\left[ \left \Vert \nabla F_k(\tilde{\mathbf{w}}_m) \right\Vert^2\right],
\end{align}
where (a) uses $\mathbb{E}\left[ \left(\frac{\chi_k\xi_{k,\text{D}}}{r_k}\right)^2\right]=\frac{1}{p_k r_k}$ and $\mathbb{E}\left[ \left(\frac{\chi_k\xi_{k,\text{D}}}{r_k}-1\right)^2\right]=\frac{1}{p_k r_k}-1$. Next, by expanding the square term, we reformulate $B_2$ as
\begin{align}\label{b2}
B_2&= \sum_{k=1}^K \alpha_k \mathbb{E}\left[\left \Vert \nabla F_k(\tilde{\mathbf{w}}_m)-\sum_{i=1}^K \alpha_i\nabla F_i(\tilde{\mathbf{w}}_m) \right \Vert^2 \right ]\nonumber \\
&=\sum_{k=1}^K \alpha_k \left( \mathbb{E}\left[\left \Vert \nabla F_k(\tilde{\mathbf{w}}_m)\right \Vert^2 \right]+\mathbb{E}\left[\left \Vert \sum_{i=1}^K \alpha_i\nabla F_i(\tilde{\mathbf{w}}_m)\right \Vert^2 \right]\right. \nonumber \\
&\quad \left.-2\mathbb{E}\left[ \nabla F_k(\tilde{\mathbf{w}}_m)^T \left( \sum_{i=1}^K \alpha_i\nabla F_i(\tilde{\mathbf{w}}_m)\right)\right] \right)\nonumber \\
&=\sum_{k=1}^K \alpha_k  \mathbb{E}\left[\left \Vert \nabla F_k(\tilde{\mathbf{w}}_m)\right \Vert^2 \right]-\mathbb{E}\left[\left \Vert \sum_{i=1}^K \alpha_i\nabla F_i(\tilde{\mathbf{w}}_m)\right \Vert^2 \right]\nonumber \\
&=\sum_{k=1}^K \alpha_k  \mathbb{E}\left[\left \Vert \nabla F_k(\tilde{\mathbf{w}}_m)\right \Vert^2 \right]-\mathbb{E}\left[\left \Vert \nabla F(\tilde{\mathbf{w}}_m)\right \Vert^2 \right].
\end{align}
Then for $A_2$, we have
\begin{align}\label{a2}
A_2& = \mathbb{E}\left [\left \Vert \tilde{\mathbf{w}}_m-\mathbf{w}^* -\eta  \nabla F (\tilde{\mathbf{w}}_m) \right \Vert^2 \right ]\nonumber \\
&=\mathbb{E}\left [ \Vert \tilde{\mathbf{w}}_m-\mathbf{w}^* \Vert^2 \right ]-2\eta \mathbb{E}\left[(\tilde{\mathbf{w}}_m-\mathbf{w}^*)^T \nabla F(\tilde{\mathbf{w}}_m) \right]\nonumber \\
&\quad + \eta^2 \mathbb{E}\left [ \Vert\nabla F(\tilde{\mathbf{w}}_m)\Vert^2 \right ] \nonumber \\
&\overset{\text{(a)}}{\leq} (1-\eta \mu) \mathbb{E}\left[\left \Vert \tilde{\mathbf{w}}_m -\mathbf{w}^*\right \Vert^2 \right]+ 2\eta \mathbb{E}\left[ F(\mathbf{w}^*)-F(\tilde{\mathbf{w}})\right ]\nonumber \\
&\quad + \eta^2 \mathbb{E}\left [ \Vert\nabla F(\tilde{\mathbf{w}}_m)\Vert^2 \right ]  \nonumber \\
&\overset{\text{(b)}}{\leq}\left( 1-\eta \mu \right)\mathbb{E}\left[\left \Vert \tilde{\mathbf{w}}_m -\mathbf{w}^*\right \Vert^2 \right]+ \eta^2 \mathbb{E}\left [ \Vert\nabla F(\tilde{\mathbf{w}}_m)\Vert^2 \right ] ,
\end{align}
where the inequality in (a) is due to \emph{Assumption 1}, and (b) is due to the fact that
$F(\mathbf{w}^*)-F(\mathbf{w})\leq 0$ for $\forall \mathbf{w}\in \mathbb{R}^d$.

Combining all the results in (\ref{a1})-(\ref{a2}), it yields 
\begin{align}\label{eq55}
\mathbb{E}& \left [ \left\Vert \tilde{\mathbf{w}}_{m+1}-\mathbf{w}^* \right \Vert^2\right]\nonumber \\
&\leq \left( 1-\eta \mu \right)\mathbb{E}\left[\left \Vert \tilde{\mathbf{w}}_m -\mathbf{w}^*\right \Vert^2 \right]\nonumber \\
&\quad +\sum_{k=1}^K  \frac{\eta^2\alpha_k}{p_k r_k}\mathbb{E}\left[ \left \Vert \nabla F_k(\tilde{\mathbf{w}}_m) \right\Vert^2\right]+\sum_{k=1}^K \frac{\eta^2\alpha_k \phi(b)}{p_k r_k}.
\end{align}
We further rewrite the second term in the right hand side (RHS) of (\ref{eq55}) as
\begin{align}
\mathbb{E}&\left[ \left \Vert \nabla F_k(\tilde{\mathbf{w}}_m) \right\Vert^2\right]\nonumber \\
&\overset{\text{(a)}}{=} \mathbb{E}\left[ \left \Vert \nabla F_k (\tilde{\mathbf{w}}_m)-\nabla F_k (\mathbf{w}_k^*)\right \Vert^2 \right ]\nonumber \\
&\overset{\text{(b)}}{\leq}  L^2  \mathbb{E}\left[ \left \Vert \tilde{\mathbf{w}}_m-\mathbf{w}_k^*\right \Vert^2 \right ]=L^2  \mathbb{E}\left[ \left \Vert \tilde{\mathbf{w}}_m-\mathbf{w}^* +\mathbf{w}^*-\mathbf{w}_k^*\right \Vert^2 \right ]\nonumber\\
&\overset{\text{(c)}}{\leq}2 L^2 \mathbb{E}\left[ \left \Vert \tilde{\mathbf{w}}_m-\mathbf{w}^* \right \Vert^2 \right ] +2L^2 \delta^2,
\end{align} 
where (a) comes from $\nabla F_k (\mathbf{w}_k^*)=\mathbf{0}$, (b) exploits  \emph{Assumption~2}, and (c) uses \emph{Assumption 3} and the inequality $\Vert \mathbf{a} +\mathbf{b} \Vert^2 \leq 2 \Vert \mathbf{a} \Vert^2+ 2\Vert \mathbf{b} \Vert^2$. By defining $g_{\text{D}}(\mathbf{r},b)=\sum_{k=1}^K \frac{\alpha_k}{p_k r_k}$, we conclude that
\begin{align}\label{ap2}
\mathbb{E}&\left[\left \Vert \tilde{\mathbf{w}}_{m+1} -\mathbf{w}^*\right \Vert^2 \right]\nonumber \\
&\leq \left( 1-\eta \mu +2\eta^2 L^2g_{\text{D}}(\mathbf{r},b)\right)\mathbb{E}\left[\left \Vert \tilde{\mathbf{w}}_m -\mathbf{w}^*\right \Vert^2 \right]\nonumber \\
&\quad  + \eta^2 (\phi(b)+2L^2\delta^2) g_{\text{D}}(\mathbf{r},b)\nonumber \\
&\leq \left( 1-\eta \mu +2\eta^2 L^2g_{\text{D}}(\mathbf{r},b)\right)^{m+1} \mathbb{E}\left[\left \Vert \tilde{\mathbf{w}}_{0} -\mathbf{w}^*\right \Vert^2 \right]\nonumber \\
&\quad+\frac{\eta (\phi(b) +2L^2\delta^2) g_{\text{D}}(\mathbf{r},b)}{\mu-2\eta L^2 g_{\text{D}}(\mathbf{r},b)}.
\end{align}
Plugging (\ref{ap2}) into (\ref{ap1}), we obtain the convergence result and complete the proof.

\section{Proof of Theorem 2}\label{theo2}
As for the analog transmission, the main difference from the digital transmission lies in the term $B_1$ in (\ref{a1}). With the analog case, $B_1$ is expressed as
\begin{align}\label{b1a}
B_1&=\sum_{k=1}^K \alpha_k \mathbb{E}\left[\left \Vert \frac{\chi_k\xi_{k,\text{A}}}{ r_k} \mathbf{g}_m^k+\bar{\mathbf{z}}_m -\mathbf{g}_{m}^k\right \Vert^2 \right ]\nonumber \\
&= \sum_{k=1}^K \alpha_k \mathbb{E}\left[ \left \Vert \left(\frac{\chi_k\xi_{k,\text{A}}}{r_k} -1\right)\mathbf{g}_m^k\right \Vert^2 \right ]+\mathbb{E}\left[\Vert \bar{\mathbf{z}}_m \Vert^2 \right]\nonumber \\
&=\sum_{k=1}^K \alpha_k \mathbb{E}\left[ \left(\frac{\chi_k\xi_{k,\text{A}}}{r_k}-1\right)^2\right]\mathbb{E}\left[ \left \Vert \nabla F_k(\tilde{\mathbf{w}}_m) \right\Vert^2\right]\nonumber \\
&\quad +\mathbb{E}\left[\Vert \bar{\mathbf{z}}_m \Vert^2 \right].
\end{align}
For the equivalent noise term, recalling that $\bar{\mathbf{z}}_m=\frac{\Re\left\{\mathbf{z}_m\right\}}{\zeta}$, we first derive the scaling factor $\zeta$. Constrained by the transmit power budget in (\ref{pmax}), the scaling factor $\zeta$ must satisfy
\begin{align}\label{a28}
\max_{k\in \mathcal{S}_m} \left\{\frac{\zeta^2 e^{2\gamma_{\mathrm{th}}} \alpha_k^2 d_k^{\alpha} }{\rho^2 r_k^2 \vert \hat{h}_k \vert^2}\mathbb{E}\left[ \left \Vert\mathbf{g}_m^k \right \Vert^2 \right] \right\} \leq P_{\max}. 
\end{align}
Based on \emph{Assumption 3} and the definition in (\ref{eq3}), we can conclude that 
\begin{align} \label{eq58}
\left \Vert \mathbf{g}_m^k \right \Vert \leq \frac{1}{D_k} \sum_{\mathbf{u}\in \mathcal{D}_k}\left \Vert \nabla \mathcal{L}(\mathbf{w}_m,\mathbf{u}) \right \Vert\leq \gamma.
\end{align}
Note that for all the activated devices, we have $\vert \hat{h}_k \vert^2\geq \gamma_{\mathrm{th}}$. Hence, we select the feasible $\zeta$ as
\begin{align}\label{zeta}
\zeta = \frac{\rho \sqrt{P_{\max}\gamma_{\mathrm{th}}}} {\gamma e^{\gamma_{\mathrm{th}}}}\min_k \left \{\frac{r_k}{\alpha_k }d_k^{-\frac{\alpha}{2}}\right\}.
\end{align}
Then, we have 
\begin{align}
\mathbb{E}\left[\Vert \bar{\mathbf{z}}_m \Vert^2 \right]=\frac{B N_0 \gamma^2 e^{2\gamma_{\mathrm{th}}}}{2P_{\max}\rho^2 \gamma_{\mathrm{th}}} \max_k\left \{\frac{\alpha_k^2}{r_k^2 }d_k^{\alpha}\right\}.
\end{align}
Next, the variance of the coefficient distortion, $\frac{\chi_k\xi_{k,\text{D}}}{r_k}$, is calculated as
\begin{align}\label{cov}
\mathbb{E}&\left[ \left(\frac{\chi_k\xi_{k,\text{A}}}{r_k}-1\right)^2\right]\overset{\text{(a)}}{=}r_k \mathbb{E}\left[ \left(\frac{\xi_{k,\text{A}}}{r_k}-1\right)^2\right]+1-r_k \nonumber \\
&=r_k \left( \mathbb{E}\left [ \left.\left(\frac{\Re\{h_k^* \hat{h}_k\}\mathrm{e}^{\gamma_{\mathrm{th}}}}{\vert \hat{h}_k \vert^2 \rho r_k }-1\right)^2 \right \vert \vert \hat{h}_k \vert^2\geq \gamma_{\mathrm{th}} \right ]\right.\nonumber \\
&\quad \left. \times\Pr\left \{ \vert \hat{h}_k \vert^2\geq \gamma_{\mathrm{th}}\right\}+\Pr\left \{ \vert \hat{h}_k \vert^2< \gamma_{\mathrm{th}}\right\}\right)+1-r_k \nonumber\\
&= r_k e^{-\gamma_{\mathrm{th}}}\mathbb{E}\left [ \left.\left(\frac{\Re\{h_k^* \hat{h}_k\}\mathrm{e}^{\gamma_{\mathrm{th}}}}{\vert \hat{h}_k \vert^2 \rho r_k}-1\right)^2 \right \vert \vert \hat{h}_k \vert^2\geq \gamma_{\mathrm{th}} \right ]\nonumber \\
&\quad + 1-r_k e^{-\gamma_{\mathrm{th}}}\nonumber \\
&\overset{\text{(b)}}{=}\left(e^{\gamma_{\mathrm{th}}}+\frac{(1-\rho^2)\mathrm{E}_1\left(\gamma_{\mathrm{th}} \right)e^{2\gamma_{\mathrm{th}}}}{2\rho^2}\right)\frac{1}{r_k}-1,
\end{align}
where (a) exploits the independence of $\chi_k$ and $\xi_{k,\text{A}}$, (b) is due to \cite[Eq.~(25)]{imperfect}. Substituting (\ref{cov}) into (\ref{b1a}) and combining the results in (\ref{b2}) and (\ref{a2}), we complete the proof.

\section{Proof of Corollary 2}\label{coll1}
To begin with, the expectation of $\vert \beta_k \vert^2$ is calculated as 
\begin{align}
\mathbb{E}\left[\vert \beta_k \vert^2 \right]&=\frac{\zeta^2 \lambda^2 \alpha_k d_k^\alpha}{r_k^2} \mathbb{E}\left[\frac{1}{\vert \hat{h}_k \vert^2}\right]\nonumber \\
&\overset{\text{(a)}}{=} \frac{\zeta^2 \lambda^2 \alpha_k d_k^\alpha }{r_k^2}\mathrm{E}_1(\gamma_{\mathrm{th}}),
\end{align}
where (a) comes from the fact that $\vert \hat{h}_k \vert^2$ follows an exponential distribution and the integral $\int_{\gamma_{\mathrm{th}}}^{\infty} \frac{1}{x}e^{-x}\mathrm{d}x = \mathrm{E}_1(\gamma_{\mathrm{th}})$. Substituting (\ref{eq58}) into (\ref{pave}), we get a feasible $\zeta$ as
\begin{align}
\zeta=\frac{\rho \sqrt{P_{\mathrm{ave}}}}{\gamma e^{\gamma_{\mathrm{th}}}\sqrt{\mathrm{E}_1(\gamma_{\mathrm{th}})}}\min_k \left \{\frac{r_k}{\alpha_k }d_k^{-\frac{\alpha}{2}}\right\}.
\end{align}
Then, following the same steps as in Appendix \ref{theo2}, we complete the proof.

\end{document}